\pdfoutput=1 

\documentclass[journal,romanappendices]{IEEEtran} 

\usepackage[tight,footnotesize]{subfigure}

\usepackage{graphicx}
\usepackage{amssymb}
\usepackage{array}
\usepackage{dsfont} 

\usepackage{cite}
\usepackage{amsmath}
\usepackage{amsthm}
\usepackage{graphicx}
\usepackage{amssymb}

\usepackage{array}
\usepackage{color}

\usepackage{multirow}
\usepackage{ifpdf}
\usepackage[T1]{fontenc}
\graphicspath{{./Figs/}}                   

\theoremstyle{definition}
\newtheorem{definition}{Definition}
\newtheorem{lemma}{Lemma}
\newtheorem{theorem}{Theorem}
\newtheorem{corollary}{Corollary}
\newtheorem{remark}{Remark}
\newtheorem{properties}{Properties}
\newtheorem{problem}{Problem}

\newcommand{\prob}{\mathbb{P}}

\newcommand{\Ex}{\mathbb{E}}
\newcommand{\ve}[1]{\boldsymbol{#1}}
\newcommand{\R}{\mathbb{R}}

\newcommand{\bigone}{\mathds{1}}

\newcommand{\tr}{\text{tr}}

\newcommand{\defeq}{\triangleq}
\newcommand{\ess}{\text{ess}}
\newcommand{\diag}{\text{diag}}

\newcommand{\Hip}{\mathcal{H}}

\begin{document}
\title{Distributed Detection of a Random Process over a Multiple Access Channel under Energy and Bandwidth Constraints}

\author{{Juan Augusto Maya,~\IEEEmembership{Member,~IEEE,} Leonardo Rey Vega,~\IEEEmembership{Member,~IEEE,} and Cecilia G. Galarza}

\thanks{The three authors are with the University of Buenos Aires (UBA), Buenos Aires, Argentina. L. Rey Vega and C. Galarza are also with the CSC-CONICET, Buenos Aires, Argentina. 
Emails:  \{jmaya, lrey, cgalar\}@fi.uba.ar} %
\thanks{This work is an extension of the results previously presented in \cite{MayaSAM} and was partially supported by the Peruilh grant (UBA) and the project UBACYT 2002010200250.} %
}



\maketitle

\begin{abstract}
We analyze a binary hypothesis testing problem built on a wireless sensor network (WSN) for detecting a stationary random process distributed both in space and time with circularly-symmetric complex Gaussian distribution under the Neyman-Pearson framework.
Using an analog scheme, the sensors transmit different linear combinations of their measurements through a  multiple access channel (MAC) to reach the fusion center (FC), whose task is to decide whether the process is present or not. 
Considering an energy constraint on each node transmission and a limited amount of channel uses, we compute the miss error exponent of the proposed scheme using Large Deviation Theory (LDT) and show that the proposed strategy is asymptotically optimal (when the number of sensors approaches to infinity) among linear orthogonal schemes.
We also show that the proposed scheme obtains significant energy saving 
in the low signal-to-noise ratio regime, which is the typical scenario of WSNs. 
Finally, a Monte Carlo simulation of a 2-dimensional process in space validates the analytical results.   
 \end{abstract}

\IEEEpeerreviewmaketitle

\section{Introduction}
\IEEEPARstart{D}{istributed}
detection based on wireless sensor networks (WSN) is a topic which has attracted great interest in recent years (see \cite{veeravalli-varshney-2012} and references therein).  A typical WSN has a large number of sensor nodes which are generally low-cost battery-powered devices with limited sensing, computing, and communication capabilities. Sensors acquire noisy measurements, perform simple data processing and propagate the information into the WSN to reach a decision about a physical phenomenon happening in the coverage area.  The data processing at each node and the propagation of the information within the network are key aspects to be considered to achieve the required performance level.

Network resources, such as energy and bandwidth, are scarce and expensive, but they are key variables when the design is focused on processing latency and detection performance. Clever data processing strategies are required to maximize performance under resources constraints. These constraints could be imposed on a node by node basis, or on the overall network. On any matter, appropriate choices of the processing strategy could largely impact on the total cost or on the life cycle of the WSN, and as such, habilitate its deployment on remote locations or not. 

\subsection{Related work}
Distributed detection theory has been much studied in the past. Starting with the seminal work of Tenney and Sandell \cite{Tenney_1981}, several results have been derived on how each node compiles the available information and communicates with the \emph{fusion center} (FC) where the final decision on the true state of nature is taken. Under this setup, digital transmission schemes, where appropriate quantization rules has to be designed has been dealt with \cite{VarshneyDistDet, TsitsiklisDescDet, Chamberland-Veeravalli-2003} (see also references therein). On the other hand, analog communication schemes were also studied in the past (see for example \cite{Veeravalli_2003, Costa_2003, ChamberlandVeeravalli2004}). For Gaussian networks, with independent and identically distributed (i.i.d.) Gaussian measurements \cite{GastparSourceChannel, Gastpar_2008}, it is known that a simple analog scheme as the scaling and transmission of the noisy measurement, is an optimal joint source-channel scheme in terms of 
quadratic distortion with power constraint in the sensors. Clearly, this can be viewed as a strong motivation, from a theoretical point of view, for further study of analog schemes for distributed detection problems.

The specific communication strategy from the sensor nodes to the FC has also been extensively studied. The simplest approach is to consider that sensors communicate with the FC using orthogonal parallel channels, like time division multiple access (TDMA) or frequency division multiple access (FDMA)\cite{Veeravalli_2003,Bianchi_2011}. Clearly, this could not be efficient for  large-scale wireless sensor networks where a large bandwidth is required for simultaneous transmissions or a large detection delay is necessary if sensors use the same bandwidth and transmit in different time slots. 

A more sophisticated approach is implemented with a multiple access channel (MAC) where sensors transmit simultaneously. In this case, the bandwidth requirement does not necessarily depend on the number of sensors. Transmission over a MAC is appealing because, for certain distributed detection problems, linear mixing of the sensor measurements, naturally performed in a MAC, could be efficiently exploited at the FC. For example, in the case of Gaussian measurements and analog transmissions, this coherent mixing could provide beamforming gains with significant impact on the performance and the energy consumed by the network. The use of a MAC channel for the problem of distributed detection has been studied in the past. In \cite{SayeedTBMA} linear mixing through a MAC is used for collaboratively computing the network-wide \emph{type} of the measurements (which are assumed to be i.i.d. given the state of nature) taken at each node. As the noise in the channel is asymptotically harmless when the number of 
sensors grows to infinity, and the type of a set of i.i.d. measurements is a \emph{sufficient statistic}, this strategy has the same optimal performance as the best \emph{centralized} scheme (see also \cite{Mergen_2007, Anandkumar2007}). The specific case in which the possible states of nature are determined by the presence or not of a deterministic signal was considered previously \cite{MayaBiasCorr, LiDai_DetSignalMAC, DasarathanTepen2014}. 

However, when the possible states of nature are the presence or absence of a random process, and the nodes use a MAC to communicate with the FC, the problem is more delicate. When the random process is i.i.d. across time and along the sensor nodes, the analog transmission of the log-likelihood ratio (LLR) of the measurement at each node is optimal in the sense that the asymptotic performance converges to that of the centralized detector. This is a  consequence of the fact that the network-wide global LLR is simply the sum of the marginal LLRs computed at each node. However, when the random process to be detected is correlated in space and/or time the situation is not so simple. In this case, the measurements taken at each node are not i.i.d. across time and space and the sum of the marginal LLRs (achieved through the coherent combining through the MAC) is clearly suboptimal, even when there is no fading \cite{CohenLLRMAC}. As the network-wide LLR is in general a non-linear function of the marginal LLRs, linear combining of the MAC 
cannot provide the optimal statistic to the FC, as in the case of i.i.d. data. This is also evident given that when each node sends only its marginal LLR, it is neglecting the correlation with measurements of other nodes. 

Problems with correlated observations could be considerably challenging \cite{willett2000good}. Design of distributed processing strategies to benefit from the correlation among data is, in general, an open problem. It is well known that signal correlation can help to improve the detection performance, specially when low quality sensor measurements are available \cite{ChamberlandVeeravalli2004,PoorLattice2D}. Clearly, a clever use of the correlation in space and/or time requires of cooperation among nodes. This problem has also been studied in previous works \cite{Sung_2007,Tong_2007}. However, the specific case of the distributed detection of a random process where the nodes communicate with the FC through a MAC still deserves some more study. 

\subsection{Contributions}

 We will consider a distributed detection scenario where all the sensors communicate with the FC through a MAC. Each sensor obtains a local measurement that is a realization of a Gaussian stochastic process arbitrarily correlated in space and possibly in time. Our goal will be to analyze the asymptotic performance of the detection scheme, when the number of sensors approaches to infinity. In particular, we look at the error exponents under the Neyman-Pearson scenario. We will consider several processing strategies that exploit correlation in space and time. These strategies can be briefly described as follows. On a first step, each sensor takes a single measurement or a set of measurements. Then, in a synchronous manner, all the sensors transmit different analog linear combinations of their measurements using several MAC uses. Finally, the FC gathers all the data  and constructs an appropriate statistic to make the decision. The cooperation among nodes is achieved through multiple channel uses. As these multiple uses are clearly pricey, we also impose that the sensors have a limited energy budget to be spent on these transmissions, and we carefully select the energy used in each channel use in order to comply with this budget. We obtain then the optimal energy allocation policy which depends on the statistical properties of the process to be detected. It is shown however that for a large number of nodes, knowledge of these statistics is needed only at the FC. 
As a side result, the optimal number of MAC uses is also obtained. It is shown, that in general the required number of  channel uses is not very large for general correlated processes. This implies that this form of cooperation does not impose severe penalties in bandwidth or delay, allowing for a close to optimal performance in terms of error exponents. The performance gains obtained with the proposed schemes are more important when the channel signal to noise ratio (SNR$_\text{C}$) is small, which is the usual scenario for WSNs designed for long life cycle. 
%
%
%
%
\subsection{Notation and Organization}
Vectors are written in boldface and matrices in capital letters. $A_n$ and $B_{nm}$ are  square and  rectangular matrices of sizes $n\times n$ and $n\times m$, respectively. $\det A_n$ and $\tr A_n$ are the determinant and the trace of $A_n$, respectively. $(\cdot)^T$ and $(\cdot)^H$ denote transpose and transpose conjugate. We do not distinguish between random variables and their realization values. $p(\ve{x}|\Hip)$ is the probability density function (p.d.f.) of $\ve{x}$ conditioned to $\Hip$. $\prob_i(\cdot)$ and $\Ex_i(\cdot)$ are the probability and the expectation, respectively, computed under hypothesis $\Hip_i$. The pre-image of the set $A$ through the function $\phi(\cdot)$ is $\phi^{-1}(A)= \{\nu\in \text{dom}(\phi) :\phi(\nu)\in A\}$. $\bigone(\nu \in A)$ is the indicator function, i.e., it is 1 if $\nu \in A$ and 0 otherwise.
The set of absolutely integrable and essentially bounded functions of support $[0,1]$ are, respectively, $L^1([0,1])=\{f: \int_{0}^{1} |f(\nu)|d\nu<\infty \}$ and $L^\infty([0,1])=\{f: \ess\sup |f(\nu)| <\infty \}$. 

The paper is organized as follows. In Section \ref{sec:DetProblem}, we formulate the binary hypothesis testing problem and describe the MAC between the sensors and the FC. In Section \ref{sec:Stats}, we present the centralized detector and then we develop two strategies for distributed detection. In Section \ref{tools}, we enunciate auxiliary tools used to compute the error exponents in Section \ref{sec:exponents}. In Section \ref{sec:results1D}, we show numerical results for one-dimensional networks and in Section \ref{sec:p-WSN} we extend  the results for multi-dimensional networks. In Section \ref{sec:Applications} we apply the results for detecting a 2D spatial random process and in Section \ref{sec:Conclusions} we elaborate on the main conclusions. Technical proofs are provided in the appendices.

\section{Detection Problem}
\label{sec:DetProblem}
We first consider a network where the nodes are distributed along a line and each sensor takes a single measurement. Networks with more dimensions (in space and/or time) will be considered in section \ref{sec:p-WSN}. The measurement at the $k$-th sensor under each hypothesis is:
\begin{equation}
\left\{
\begin{array}{llll}
\Hip_1: & x_k&=s_k + v_k,& \\
\Hip_0: & x_k&=v_k,& k=1,\dots,n.
\end{array}
\right.
\label{eq:modelo}
\end{equation} 
We assume that $s_k$ is a zero-mean circularly-symmetric complex Gaussian stationary process with variance $\sigma^2_s$ and power spectral density (PSD) $\phi(\nu)$, and $v_k$ is a zero-mean circularly-symmetric complex white Gaussian noise independent of $s_k$ with variance $\sigma_v^2$. Thus, $x_k$ is Gaussian distributed either under $\Hip_0$ or $\Hip_1$. We define the following column vectors: $\ve{s}=[s_1,\dots,s_n]^T$, $\ve{v}=[v_1,\dots, v_n]^T$ and $\ve{x}=[x_1,\dots, x_n]^T$.
The covariance matrix of $\ve{v}$ is $\sigma_v^2 I_n$ where $I_n$ is the identity matrix of dimension $n$. The signal vector $\ve{s}$ has a Toeplitz covariance matrix $\Sigma_n\defeq\Sigma_n(\phi)$ whose $(i,j)$-th element is completely characterized by 
$\phi(\nu)$ as,
\begin{equation*}
(\Sigma_n)_{i,j}=\int_0^1 \phi(\nu)e^{-\jmath 2\pi\nu(i-j)}d\nu\qquad 1\leq i,j\leq n.
\end{equation*}
where $\!\nu\!$ is the normalized frequency. 
The covariance matrices of $\!\ve{x}\!$ are $\Sigma_{0,n}\!\!=\!\!\sigma_v^2 I_n$ under $\Hip_0$, and $\Sigma_{1,n}\!\!=\!\!\Sigma_n\!+\!\sigma_v^2 I_n$ under $\Hip_1$. 

\begin{remark}
The Toeplitz assumption allows us to manage the correlation between the measurements at the nodes in a simple manner when the number of them grows unbounded. It has also been considered in \cite{LiDai_DetSignalMAC,ChamberlandVeeravalli2004}. Physically, it can be linked to a situation in which the sensor nodes are located on a regular grid and the continuous random process in space is stationary. More general  correlation models include Gaussian random fields \cite{Sung_2007,Tong_2007}. However, under the setup considered in this paper, 
interesting conclusions and closed form results would more involved. For that reason, we will work with processes described by Toeplitz covariance matrices, which allows us to consider sufficiently general correlation models.
\end{remark}

We analyze a sensor network where the nodes communicate with the FC through a MAC with equal gain on each link node-FC. 
This is a simplification for the general setup, however it is a reasonable approximation for networks
 deployed in rural and remote areas. In this case, each node has a line of sight with the FC, the nodes are steady, and the surroundings do not vary much. Under this scenario, node synchronization and channel inversion at the nodes are feasible. This assumption has been extensively considered in the past \cite{Tepeden2010DistEstMAC,Tepeden2012DistEstMAC, Tepen_2014,Goldenbaum2010ComputingMAC,SayeedTBMA,Evans2008DetMAC}. 

Consider now $n'$ channel uses\footnote{We will refer to channel uses or degrees of freedom (DoF) interchangeably.}, with $n'\leq n$. Each channel use is associated to either an orthogonal time slot or a frequency band. Then, the processing strategy may use $n'$ time slots, $n'$ frequency bands or a combination of them such that the product of time slots and frequency bands is $n'$. 
During a channel use, sensors communicate with the FC through a noisy MAC without any other interference. 
The signal collected at the FC is a noisy version of the coherent superposition of the symbols transmitted by the $n$ sensors through the MAC,
\begin{equation}
z_{k'}=\sum_{k=1}^{n} g_{k k'}(x_k) + w_{k'},\ \ k'=1,\dots,n', 
\label{zFC1}
\end{equation}
where $g_{k k'}(\cdot)$ is the encoding function used by the sensor $k$ at channel use $k'$, and $w_{k'}$ is the zero-mean communication noise with variance $\sigma_w^2$ and circularly-symmetric complex Gaussian distribution independent of everything else. An illustration of the distributed scheme is shown in Fig. \ref{fig:wsnDD}. 
\begin{figure}[hbt]
\centering
\includegraphics[width=\linewidth]{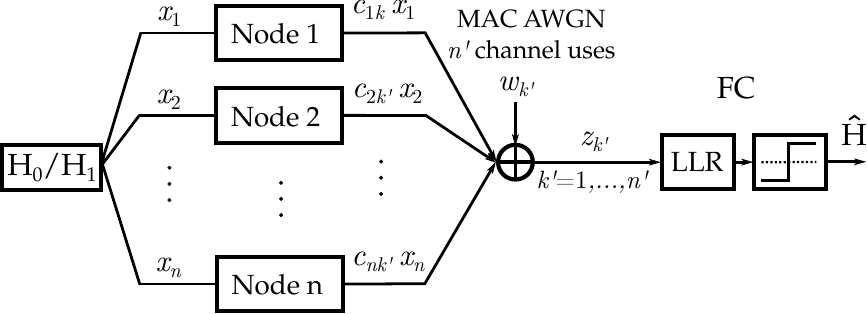}
\caption{Decentralized detection scheme for a WSN with linear encoding functions: $g_{kk'}(x_k)=c_{kk'} x_k$.}
\label{fig:wsnDD}
\end{figure}

\section{Detectors}
\label{sec:Stats}
\subsection{Centralized detector}  
Suppose that the FC has direct access to the complete measured vector $\ve{x}$ through $n$ orthogonal noiseless communication channels. The appropriate use of these measurements allows the FC to construct the optimal centralized detector (CD) \cite{Kay_SSP}. In a distributed setting, with noisy links from the sensors to the FC, this detector provides an upper bound (not necessarily tight) on the performance of any distributed scheme.
Consider the Neyman-Pearson problem for a fixed false alarm probability level $\alpha$, where a false alarm event occurs when $\Hip_1$ is declared but $\Hip_0$ is true. The associated normalized logarithmic likelihood ratio (LLR) is \cite{Kay_SSP}
\begin{align}
T_\text{c}^{n}(\ve{x}) &=\frac{1}{n} \log \frac{p(\ve{x}|\Hip_1)}{p(\ve{x}|\Hip_0)}\nonumber\\
&= \frac{1}{n}\left(\ve{x}^H \left(\Sigma_{0,n}^{-1}-\Sigma_{1,n}^{-1} \right)\ve{x}-  \log\frac{\det\Sigma_{1,n}}{\det\Sigma_{0,n}}\right).
\label{TC1}
\end{align} 
Now, according to the Neyman-Pearson theorem,  the optimum centralized test chooses $\Hip_1$ if $T_c^n(\ve{x}_n) > \tau_n$, and $\Hip_0$ otherwise, where the threshold of the test $\tau_n$ depends on $\alpha$. 

When the process $\{s_k\}$ is spatially correlated, the covariance matrix $\Sigma_{1,n}$ is not diagonal. Hence, (\ref{TC1}) cannot be expressed as the sum of the marginal LLRs from each node, and the mixing property of the MAC does not lead to the global LLR as in \cite{CohenLLRMAC}.
However, it is possible to implement simple distributed detection schemes to reconstruct, in some way, the centralized statistic at the FC through multiple channel uses.
\subsection{Distributed detector}
We consider distributed detection (DD) schemes where the nodes make $n'$ channel uses through a MAC.
In this work, we restrict ourselves to the case where $g_{kk'}(\cdot)$ are linear encoding functions. Thus, in the case of one-dimensional networks, each sensor transmits scaled versions of its local measurement through the MAC and (\ref{zFC1}) results in: 
\begin{equation*}
z_{k'}=\ve{c}_{k'}^H\ve{x}+w_{k'},\qquad k'=1,\dots,n', 
\label{zFC}
\end{equation*}
where $\boldsymbol c_{k'}=[c_{1k'}, c_{2k'} ,\dots, c_{nk'}]^T$. 
Define the \emph{precoding matrix} as $C_{nn'}=[\boldsymbol c_{1}, \dots,\boldsymbol c_{n'}]$ and $\ve{z}=[z_{1},\dots,z_{n'}]^T$. Then
\begin{equation}
\label{eq:zvec}
\ve{z}=C_{nn'}^H \ve{x} + \ve{w},
\end{equation}
where $\ve{w}=[w_{1},\dots,w_{n'}]^T$ has covariance matrix $\sigma^2_w I_{n'}$. Considering the Neyman-Pearson problem, the normalized LLR at the FC is:
\begin{align}
T_\text{d}^{n}(\ve{z}) &= \frac{1}{n}\log \frac{p(\ve{z}|\Hip_1)}{p(\ve{z}|\Hip_0)}\nonumber\\
&= \frac{1}{n}\left(\ve{z}^H \left(\Xi_{0,n'}^{-1}-\Xi_{1,n'}^{-1} \right)\ve{z} -  \log\frac{\det\Xi_{1,n'}}{\det\Xi_{0,n'}}\right),
\label{TD1}
\end{align} 
where $\Xi_{1,n'}= C_{nn'}^H\Sigma_n C_{nn'} + \sigma^2_v C_{nn'}^HC_{nn'}+ \sigma^2_w I_{n'}$, and $\Xi_{0,n'}=\sigma^2_v C_{nn'}^H C_{nn'} + \sigma^2_w I_{n'}$ 
are the covariance matrices of $\ve{z}$ under and $\Hip_1$ and $\Hip_0$, respectively.
The false alarm and the miss error probability are defined as $P_{fa}^n=\prob_0(T^{n}_d(\ve{z})>\tau_n)$ and $P_m^n=\prob_1(T^{n}_d(\ve{z})<\tau_n)$, respectively. 
%
The average energy consumed by the $k$-th node during $n'$ transmissions is:
\begin{equation}
E_{k}= \mathbb{E}\left[\sum_{k'=1}^{n'}|c_{kk'}x_k|^2 \right]\!\!=\mathbb{E}[|x_k|^2]\sum_{k'=1}^{n'}|c_{kk'}|^2,\, k=1,\dots,n, \nonumber
\label{eq:energy_constrI}
\end{equation}
where $c_{kk'}$ is the $(k,k')$-th element of $C_{nn'}$, $\mathbb{E}[|x_k|^2]=\sigma_v^2+p_1\sigma_s^2$, and $p_1$ is the \emph{a priori} probability of the state of nature $\Hip_1$. 
When $p_1$ is unknown, a natural upper-bound for $E_k$ is obtained by taking $p_1=1$.

The problem is then to obtain appropriate precoding matrices $C_{nn'}$ such that the constraints on energy and channel uses are satisfied in the Neyman-Pearson setting: 
\begin{problem}[Best linear precoding strategy]
Consider that $n$ sensors take measurements according to (\ref{eq:modelo}) and the communication model is (\ref{eq:zvec}). When the per-sensor average energy constraint is $E_t$, and the probability of false alarm is less than $\alpha$, the best linear precoding strategy is obtained by solving 
\begin{gather}
\inf_{C_{nn'}\in\mathcal{A}_{nn'}} P_m^n, \nonumber
\label{eq:optimal_prob}
\end{gather}
with 
$\mathcal{A}_{nn'}\!\!=\!\{C_{nn'}\!\!\in\mathbb{C}^{n\times n'}\!\!\!:\! P_{fa}^n\!\!\leq\!\alpha, \frac{1}{n}\tr(C_{nn'}C_{nn'}^H )\!\leq\! \frac{E_{t}}{\sigma_s^2+\sigma_v^2}\}$.
\end{problem}

This problem is not convex in $C_{nn'}$, and closed-form solutions are not readily available. Therefore, we restrict the precoding matrices $C_{nn'}$ to have orthogonal columns and consider WSNs with large number of nodes where 
it is relevant to compute the \emph{error exponent} of $P_m^n$ when $P_{fa}^n\leq \alpha$. Then, the problem that we attack is as follows:
\begin{problem}[Best asymptotic orthogonal precoding strategy]
\label{best_asymp_prob}
Consider that $n$ sensors take measurements according to (\ref{eq:modelo}) and they communicate with the FC through a MAC as in (\ref{eq:zvec}). Consider also that $\frac{n'}{n}\rightarrow \beta$ when  $n\rightarrow\infty$ for a given asymptotic fraction of DoF $\beta\in(0,1]$. When the per-sensor average energy constraint is $E_t$, and the level of false alarm is limited to $\alpha$, the best asymptotic orthogonal precoding strategy is obtained by solving
\begin{gather}
\sup_{\{C_{nn'}\}_{n=1}^\infty \in \{\mathcal{A}_{nn'}^\text{orth}\}_{n=1}^\infty }
-\lim_{n\rightarrow\infty}\frac{1}{n}\log P_m^n
\nonumber
\label{eq:asymp_optimal_prob}
\end{gather}
with $\mathcal{A}_{nn'}^{\text{orth}} = \{C_{nn'}\!\!\in\mathbb{C}^{n\times n'}\!\!: C_{nn'}=V_{nn'}\Delta_{n'},\ V_{nn'}^H V_{nn'} = I_{n'}, \Delta_{n'}=\diag(\gamma_{1}^n,\dots,\gamma_{n'}^n),\ P_{fa}^n\leq\alpha, \frac{1}{n}\sum_{k'=1}^{n'} (\gamma_{k'}^n)^2\leq \frac{E_{t}}{\sigma_s^2+\sigma_v^2}\}$, where 
the coefficients $\{\gamma_{k'}^n\}$ control the energy transmitted by the sensors on each channel use.  
\end{problem}

To tackle this problem, we consider two different parameterizations for the sequence of precoding matrices $\{C_{nn'}\}$. Before elaborating on them, we make the following definition:
\begin{definition}[Transmitted Modes Set]
\label{def:freqSet}
Let $\text{L}$ be the uniform probability measure defined on the interval $[0,1]$, 
and denote by $\text{L}[\phi^{-1}(A)]$ the measure of the set $\phi^{-1}(A)$. Define the complementary cumulate distribution function $\Omega(t)$ as the measure of the set of frequencies $\nu$ such that $\phi(\nu)$ is at least $t$, i.e., $\Omega(t)=\text{L}[\phi^{-1}((t,+\infty))]$. For $\beta\in(0,1]$, define 
$\Theta_\beta:=\{\nu\in[0,1]:\Omega(\phi(\nu))\leq \beta\}$ as the transmitted modes set of Lebesgue measure $\beta$ whose elements are the frequencies $\nu$ that take on the largest values of $\phi(\nu)$. 
\end{definition}


Now, to solve Problem \ref{best_asymp_prob}, we consider two strategies: 

\begin{definition}[Principal Component Strategy in a MAC channel, PCS-MAC]
\label{PCS-MAC}
Let $n\in\mathbb{N}$ and $n' \defeq n'(n)\leq n$. Assume that $\{\ve{u}_{k}\}_{k=1}^n$ is a basis of eigenvectors for $\Sigma_n$ and $\{\lambda_k^n\}_{k=1}^n$ the corresponding eigenvalues, where $\lambda_1^n\geq\lambda_2^n\geq \cdots\geq\lambda_n^n$. We choose the precoding matrix for the PCS-MAC strategy as 
\begin{equation}
C_{nn'}= U_{n n'} \Delta_{n'},
\label{PCS-matrix}
\end{equation}
where $U_{nn'}=[\ve{u}_{1},\dots,\ve{u}_{n'}]$, $\Delta_{n'}= \diag(\gamma_{1}^n,\dots,\gamma_{n'}^n)$ and 
$n'=\max\{k\in[1,n]: \Omega(\lambda_k^n)\leq \beta\}$.
\end{definition}
To implement this strategy, the rows of $C_{nn'}$ are required to be known at the corresponding nodes. For that, either each node should perform a local eigenvalue-eigenvector decomposition, or the FC should communicate the corresponding row of $C_{nn'}$ to each sensor through a feedback channel.
The first option entails a more expensive node deployment, and the second option involves the implementation of a high throughput multi-terminal communication link whose complexity would grow linearly as the number of nodes increases. 
To reduce complexity, 
we replace the basis of eigenvectors in PCS-MAC by the Fourier basis, and we propose the next strategy: 
\begin{definition}[Principal Frequencies Strategy in a MAC channel, PFS-MAC]
\label{PFS-MAC}
Consider the power spectral density (PSD) $\phi(\nu)$ and let $n\in\mathbb{N}$ and $n'\defeq n'(n)\leq n$. For each $n$, let $(j_1,j_2,\cdots,j_n)$ be a permutation of $\{1,2,\dots,n\}$ such that $\phi\left(\frac{ j_1-1}{n}\right)\geq\phi\left(\frac{ j_2-1}{n}\right)\geq \cdots \geq \phi\left(\frac{ j_n-1}{n}\right)$.  We choose the precoding matrix PFS-MAC strategy as
\begin{equation}
C_{nn'}=F_{nn'}\Delta_{n'}
\label{PFS-matrix}
\end{equation}
where $n'=\max\left\{k\!\in\![1,n]\!:  \!\Omega\left(\phi\left(({j_k^n-1})/{n}\right)\right) \leq \beta\right\}$, $F_{nn'}=[\ve{f}_{j_1},\dots,\ve{f}_{j_{n'}}]$ is a sub-matrix of the DFT matrix of order $n$, i.e.,
$\ve{f}_{k'}=[f_{1k'},f_{2k'},\dots, f_{nk'}]^T$ with $f_{kk'} = \frac{1}{\sqrt{n}}\exp(\jmath 2\pi (k-1)(k'-1)/n)$ for $k=1,\dots,n$, $k'=j_1,\dots,j_{n'}$ and $\Delta_{n'}= \diag(\gamma_{1}^n,\cdots,\gamma_{n'}^n)$.
\end{definition}
Under this strategy, at each channel use, the FC receives a noisy version of a given frequency bin of the DFT of the measurement vector. For that, the nodes are required to know the index number of the DFT bin to be transmitted only. This information is common to all nodes and it may be broadcasted by the FC on a low rate feedback channel.


Let us interpret the above definitions. For the proposed decentralized strategies, we have imposed that the number of channel uses be limited to $n'\leq n$. 
Using a MAC in each channel use, we can limit the number of channel uses without discarding any measurement of the sensors and distributing the transmitted energy more efficiently. 
We will see in Corollary \ref{cor:DDMiss} that PCS-MAC and PFS-MAC results in the the asymptotic transmission of the components of the spectrum $\phi(\nu)$ in the set $\Theta_\beta$. Intuitively, we see that this is the best we can do, as the large components of $\phi(\nu)$ are the \emph{more informative} about the state of nature $\mathcal{H}_1$, when the fraction of channel uses approaches to $\beta$. 

\begin{remark}
For PCS-MAC, $c_{kk'}=u_{kk'} \gamma_{k'}$ where $u_{kk'}$ is the $(k,k')$-th element of $U_{nn'}$, and the average energy consumed in each node is $E_k=(p_1\sigma_s^2 + \sigma_v^2)\sum_{k'=1}^{n'} |u_{kk'}|^2 (\gamma_{k'}^n)^2$. This depends on each sensor $k$. If we additionally consider the average over all the sensors, we obtain the energy constraint considered in Problem \ref{best_asymp_prob}:
\begin{equation}
\frac{1}{n}\sum_{k'=1}^{n'} (\gamma_{k'}^n)^2\leq \frac{E_{t}}{\sigma_s^2+\sigma_v^2}.\label{eq:trConst}
\end{equation}
On the other hand, in the case of PFS-MAC , $c_{kk'}=\frac{1}{\sqrt{n}} \gamma_{k'}^n e^{\jmath 2\pi (k-1)(k'-1)/n}$, and the average energy consumed in each node is $E_k=(p_1\sigma_s^2 + \sigma_v^2)\frac{1}{n} \sum_{k'=1}^{n'}(\gamma_{k'}^n)^2$. This is independent of the sensor $k$ and therefore (\ref{eq:trConst}) is actually an energy  constraint on \emph{each} single sensor. 
\end{remark} 
It is convenient to express $\gamma_{k'}^n$ as the sampled version of a function $\xi:[0,1]\rightarrow \R_{+}$, 
\begin{equation}
 \gamma_{k'}^n=\sqrt{\xi\left(\frac{k'-1}{n}\right)},\ \ k'=1,\dots,n',
\label{eq:gamma}
\end{equation}
where  $\xi(\cdot)$ is a Riemann integrable function defined as the asymptotic energy profile.
The average energy constraint when $n$ goes to infinity results in 
\begin{align}
\lim_{n\rightarrow\infty}\frac{1}{n}\sum_{k=1}^{n'}(\gamma_k^n)^2 
&= \int_{\Theta_\beta}\xi(\nu)d\nu \leq \frac{E_t}{\sigma^2_s+\sigma^2_v}.
\label{xiConstr}
\end{align}

\begin{remark}
A similar general setup was previously considered in \cite{Bianchi_2011}. However, the authors imposed certain restrictions to the overall model that were lessen in our setup. For instance, the communication channels between the nodes and the FC were orthogonal, instead of a MAC as it is analyzed here. It was also assumed that each sensor measured the same realization of a process under $\Hip_1$ disturbed by different noise realizations. That would have been the case of a random process maximally correlated in space. Our analysis allows for general correlation functions both in space and/or time. In addition and more importantly, by introducing the coefficients $\{ \gamma_{k'}^n\}$, we optimize the energy profile for a given energy budget and we find the optimum number of MAC uses to achieve the best error exponent among the orthogonal strategies.

\end{remark} 

\section{Preliminary Tools}
\label{tools}
\begin{definition}[Weak and Strong Norms]
Let $A$ be a Hermitian $n\times n$ matrix with eigenvalues $\{\lambda_k\}_{k=1}^n$, its weak (normalized Frobenius) and strong (spectral) norms are, respectively,
$$|A|=\left(\frac{1}{n} \sum_{i=1}^{n}\sum_{j=1}^{n} |a_{ij}|^2\right)^{\!\!1/2},\ \ \|A\|=\max_{1\leq k\leq n} |\lambda_k|.$$ 
\end{definition}
\begin{definition}[Wiener Class Functions]
A function $\phi(\nu)$ defined on the normalized frequency interval $[0,1]$ is said to be in the Wiener class if it
has a Fourier series with absolutely summable Fourier coefficients $a_k$, i.e.,  $\sum_{k=-\infty}^{\infty} |a_k|<\infty$ and 
\begin{equation}
\phi(\nu)=\sum_{k=-\infty}^{\infty} a_k e^{\jmath 2\pi \nu k},\, a_k=\int_{0}^{1} \phi(\nu) e^{-\jmath 2\pi \nu k}d\nu.
\label{eq:ParFourier}
\end{equation}
\end{definition}
\begin{definition}
\label{def:circ}
The circulant matrix $B_n\defeq B_n(\phi)$ generated by the samples of the function $\phi(\nu)$ with $\nu= \frac{i-1}{n}$, $i=1,\dots,n$, is completely specified by its first row $\ve{b}_n=[b_1^n,\dots,b_n^n]$, 
\begin{equation}
B_n(\phi)=\begin{bmatrix}
b_1^n & b_2^n & \dots  & b_n^n\\ 
b_n^n & b_1^n & \ddots & b_{n-1}^n  \\ 
 	  &\ddots & \ddots &    \\ 
b_{2}^n & b_3^n & \dots  & b_1^n
\end{bmatrix}
\label{eq:CircMat}
\end{equation}
where $b_k^n=\sum_{i=1}^{n} \phi({\scriptstyle\frac{i-1}{n}}) e^{\jmath \frac{2\pi}{n} (i-1) (k-1)},\, k=1,\dots,n.$
The eigenvalues of this matrix are given by $\left\{\phi({\scriptstyle\frac{i-1}{n}})\right\}_{i=1}^{n}$.
\end{definition}
\begin{definition}[Asymptotically Equivalent Matrices]
Two sequences of $n\times n$ matrices $\{A_n\}$ and $\{B_n\}$ are said to be asymptotically equivalent, $A_n\sim B_n$, if
\begin{enumerate}
\item $A_n$ and $B_n$ are uniformly bounded in strong norm, i.e.:
\begin{equation}
\|A_n\|,\|B_n\|\leq M < \infty, \qquad n=1,2,\dots
\end{equation}
\item $A_n - B_n$ converges to zero in weak norm as $n\rightarrow\infty$, i.e.:
\begin{equation}
\lim_{n\rightarrow\infty} |A_n - B_n|=0.
\end{equation}
\end{enumerate}
\end{definition}
\begin{lemma}[Asymptotically Equivalence of Toeplitz and Circulant Matrices]
\label{lem:equiv}
Let $A_n(\phi)=[a_{i-j}]_{1\leq i,j\leq n}$ be a Toeplitz matrix with the function $\phi(\nu)$ in the Wiener class related to $a_k$ as in (\ref{eq:ParFourier}) and let $B_n(\phi)$ be a circulant matrix as in (\ref{eq:CircMat}). Then, $A_n(\phi)$ and $B_n(\phi)$ are asymptotically equivalent. 
\end{lemma}
\begin{IEEEproof}
See \cite[p. 53]{GrayToeplitz}.
\end{IEEEproof}

\begin{theorem}[Asymptotic Toeplitz Distribution Theorem]
\label{ToepTheo}
Assume that $\phi(\nu)$ is a Wiener class function and $\Sigma_{n}$ the Hermitian Toeplitz matrix related to $\phi$. Let $\{\lambda_{k}^n\}_{k=1}^{n}$ be the eigenvalues of $\Sigma_{n}$, and $\delta_2 \geq \lambda_1^n \geq \cdots \geq \lambda_n^n \geq \delta_1$.
 If $F(\cdot)$ is a continuous function defined on the interval $[\delta_1 , \delta_2]$, such that $\int_{\nu:\phi(\nu)=\delta}F(\phi(\nu))d\nu =0$ for any $\delta\in [\delta_1 , \delta_2]\cap \Delta$ where $\Delta\subseteq [0,\infty)$, then
\begin{equation*}
\lim_{{n}\rightarrow\infty}\
\frac{1}{n}\sum_{k=1}^{n} F(\lambda_{k}^n )\bigone (\lambda_k^n\in \Delta) = \int_{\phi^{-1}(\Delta)} F(\phi(\nu))d\nu.
\end{equation*}
\end{theorem}
\begin{IEEEproof}
The proof follows from \cite[Corollary 4.1]{GrayToeplitz}. 
\end{IEEEproof}

\begin{definition}[Large Deviation Principle, LDP]
\label{def:LDP}
If $G^\circ $ and $\bar{G}$ are the interior and closure of a set $G\subset\R$, respectively, we say that the sequence of random variables $\{Y_n\}_{n=1}^\infty$ satisfies the LDP with rate function $\Lambda^*(x)$ if, for any $G\subset\R$ we have  
\begin{align}
&-\inf_{x\in G^o} \Lambda^*(x) \leq \lim\inf_{n\rightarrow\infty}\frac{1}{n} \log \prob(Y_n\in G)\nonumber\\
&\leq \lim\sup_{n\rightarrow\infty}\frac{1}{n} \log \prob(Y_n\in G)\leq -\inf_{x\in \bar{G}} \Lambda^*(x).
\label{eq:LDP}
\end{align}
\end{definition}
The set $G$ is said to satisfy the $\Lambda$-continuous property if $\inf_{x\in G^\circ}\Lambda^*(x) =\inf_{x\in \bar{G}}\Lambda^*(x)$. Thus, the lower and upper bounds in (\ref{eq:LDP}) coincide, and the exponent $\epsilon_G$ is defined as:
\begin{equation}
\epsilon_G= \lim_{n\rightarrow\infty}-\frac{1}{n} \log \prob(
Y_n\in G) = \inf_{x\in G^o} \Lambda^*(x)=\inf_{x\in \bar{G}} \Lambda^*(x).
\label{LDT4}
\end{equation}
The Gärtner-Ellis theorem \cite{DemboLDT} is a handy result for the computation of a good rate function for a general sequence of random variables. For this result to be true, certain conditions on the asymptotic logarithmic moment-{gen\-er\-at\-ing} function (LMGF), defined as the limit $\Lambda(t) =\lim_{n\rightarrow\infty} \frac{1}{n}\Lambda^n(nt)$, with $\Lambda^{n}(t) = \log \Ex[e^{tY_n}]$ need to be satisfied. A particular critical condition is the steepness of $\Lambda(t)$ at the boundary of its domain. When $Y_n$ is a sequence of Gaussian quadratic forms constructed from a stationary Gaussian random process (which is the case for $T_\text{d}^{n}(\ve{z})$ and $T_\text{c}^{n}(\ve{x})$) this steepness condition is a delicate issue \cite{Bercu1997QuadraticGaussian}. The asymptotic \emph{bad behavior} of some eigenvalues of the corresponding Toeplitz matrices of the stationary process could have a critical role in the behavior of $\Lambda(t)$. For that reason, and for the particular case of the hypothesis testing problem based on likelihood ratio test, we need the following result:
\begin{theorem}[Modified Gärtner-Ellis theorem for Gaussian LLR]
\label{GETheo}
Let $\{T^n\}_{n=1}^{\infty}$ be a sequence of LLRs of complex circularly symmetric Gaussian random variables drawn from a stationary Gaussian random process with spectral density $h_i(\nu)$ under $\Hip_i$, $i=0,1$, and define its LMGF as 
\begin{equation}
\Lambda_i(t) =\lim_{n\rightarrow\infty} \frac{1}{n}\Lambda_i^n(nt), \ \ \Lambda_i^{n}(t) = \log \Ex_i[e^{tT^n}].
\label{LDT_H0}
\end{equation}
Consider the following assumptions:
\begin{itemize}
\item[$A_1)$] $h_i(\nu)$ is in the Szegö class, i.e., $\log h_i(\nu)\in L^1([0,1])$, $i=0,1$. 
\item[$A_2)$] The ratio of spectral densities is essentially bounded, i.e., $\frac{h_i(\nu)}{h_{\setminus i}(\nu)}\in L^\infty([0,1])$, where $\setminus i=1$ if $i=0$ and $\setminus i=0$ if $i=1$.
\end{itemize}
Then, under $\Hip_i$, the sequence $\{T^n\}$ satisfies the LDP in Def. \ref{def:LDP} with good rate function $\Lambda_i^*(x)$ computed through the Fenchel-Legendre transform of $\Lambda_i(t)$, $i=0,1$:
\begin{align}
\Lambda_i^*(x) = \sup_{t\in\R} \{xt-\Lambda_i(t)\},
\label{FLT:Hi}
\end{align}
with $\Lambda_i(t)=-{\displaystyle\int_{0}^1} \!\left\{\log\left(1+t \frac{h_{0}(\nu)-h_1(\nu)}{h_{\setminus i}(\nu)} \right) +t\log\frac{h_1(\nu)}{h_{0}(\nu)}\right\}d\nu.
$
\end{theorem}
\begin{IEEEproof}
See \cite[Proposition 7]{Bercu1997QuadraticGaussian}.
\end{IEEEproof}

\section{Error Exponents}
\label{sec:exponents}
We are now ready to compute the error exponents for both  the centralized detector and the decentralized detector considering both strategies, PCS-MAC and PFS-MAC. Before proceeding, we formulate some handy definitions. 
Consider a spectral density $\Gamma\defeq \Gamma(\nu)$, a set of frequencies $\Theta$, and a threshold $\tau$. Then, we define the following functionals:
\begin{align}
m_0(\Gamma) &=\int_{\Theta}\left\{\frac{\Gamma(\nu)}{1+\Gamma(\nu)} -\log(1+\Gamma(\nu))\right\}d\nu,\label{eq:m0}\\
m_1(\Gamma) &=\int_{\Theta}\{ \Gamma(\nu) -\log(1+\Gamma(\nu))\}d\nu.\label{eq:m1}
\end{align}
\begin{align}
\kappa_{fa}(\Gamma)&= \int_{\Theta}\{ \log(1-t^*\Gamma(\nu)) + t^*\log(1+\Gamma(\nu))\}d\nu\nonumber\\
&+ \tau(1+t^*),\nonumber\\ 
\kappa_{m}(\Gamma)&= \int_{\Theta}\{ \log(1-t^*\Gamma(\nu)) + t^*\log(1+\Gamma(\nu))\}d\nu\nonumber\\
&+\tau t^*,\nonumber
\end{align}
where $t^*$ is the unique solution to
\begin{equation}
\label{eq:t}
\tau +  \int_{\Theta}\log(1+\Gamma(\nu))d\nu = \int_{\Theta}\frac{\Gamma(\nu)}{1-t^{*}\,\Gamma(\nu)}d\nu.
\end{equation}

\subsection{Centralized Error Exponents}
\begin{theorem}[CD Error Exponents]
\label{CDexp}
Consider the LLR test in (\ref{TC1}) with a fixed threshold $\tau$, the spectral density $\Gamma_\text{CD}\defeq \Gamma_\text{CD}(\nu)= \frac{\phi(\nu)}{\sigma_v^2}$, $\Theta=[0,1]$ and assume that $\phi(\nu)$ is in the Wiener class. Then, the false alarm and miss error exponents for the hypothesis testing problem in (\ref{eq:modelo}) are:
\begin{equation}
\label{CDfa}
\kappa_{fa}^\text{CD}= 
\left\{
\begin{array}{cl}
\kappa_{fa}(\Gamma_\text{CD}) & \mbox{if } \tau>m_0(\Gamma_\text{CD}),\\
0 &  \mbox{if } \tau\leq m_0(\Gamma_{\text{CD}}),
\end{array}\right.
\end{equation}
\begin{equation}
\label{CDm}
\kappa_{m}^\text{CD}= 
\left\{
\begin{array}{cl}
\kappa_{m}(\Gamma_\text{CD}), & \mbox{if } \tau<m_1(\Gamma_\text{CD}),\\
0 &  \mbox{if } \tau\geq m_1(\Gamma_{\text{CD}}),
\end{array}
\right.
\end{equation}
\end{theorem}
\begin{proof}
See App. \ref{App:ProofCD}.
\end{proof}

\begin{corollary}
Consider the LLR test in (\ref{TC1}). The miss error exponent subject to $P_{fa}^n\leq \alpha$ with $\alpha\in(0,1)$ is
\begin{equation}
\kappa_{m,\alpha}^\text{CD}\!=\!\int_0^1 \!\!\left\{\frac{1}{1+\Gamma_\text{CD}(\nu)} + \log(1+\Gamma_\text{CD}(\nu)) -1 \right\}d\nu.
\label{CDmalpha}
\end{equation}
\end{corollary}
\begin{proof}
Evaluate Th. \ref{CDexp} with $\tau=m_0(\Gamma_\text{CD})+\epsilon$, where $\epsilon>0$ is arbitrary small. See \cite[Prop. 2]{PoorFreqDet} for a detailed proof. 
\end{proof}

\subsection{Decentralized Error Exponents}


Consider the setup in (\ref{eq:zvec}), and the LLR test in (\ref{TD1}). Bearing in mind the asymptotic behavior of the covariance matrices of $\ve{z}$ under $\Hip_0$ and $\Hip_1$, we formulate the following lemma:
\begin{lemma}[Decentralized Hypothesis Testing]
\label{lem:HT-DD}
Under either PCS-MAC or PFS-MAC, the decision at the FC asymptotically consists on choosing one of the following PSDs, with  $\{\gamma_{k'}^n\}$ as in (\ref{eq:gamma}), and $\nu\in\Theta_\beta$:  
\begin{equation}
\left\{
\begin{array}{llll}
\Hip_1: & h(\nu)=(\phi(\nu)+\sigma_v^2)\xi(\nu) + \sigma^2_w\\
\Hip_0: & h(\nu)= \sigma_v^2\xi(\nu) + \sigma^2_w.
\end{array}
\right.
\label{eq:testFC}
\end{equation}
\end{lemma}
\begin{proof}
See App. \ref{App:HT-DD}
\end{proof}
The following theorem establishes the error exponents for the decentralized strategies presented in Section \ref{sec:Stats}. 
Both strategies PCS-MAC and PFS-MAC allow to improve the performance of detection by careful selection of $\xi(\nu)$.
\begin{theorem}[DD Error Exponents]
\label{DDexp}
Consider the hypothesis testing problem in (\ref{eq:testFC}), 
where $\xi(\nu)$ is a fixed asymptotic energy profile that satisfies (\ref{xiConstr}). Consider also that the asymptotic fraction of DoF is limited to $\beta$. Assume that $\phi(\nu)$ is in the Wiener class and that the threshold of the test is fixed to $\tau$. Let $\Gamma_\text{DD}\defeq \Gamma_\text{DD}(\nu)= \frac{\xi(\nu)\phi(\nu)}{\xi(\nu)\sigma_v^2 + \sigma^2_w}$. Then, the strategies PCS-MAC and PFS-MAC have the same false alarm and miss error exponents,
and they are:
\begin{equation}
\label{DDfa}
\kappa_{fa}^\text{DD}= 
\left\{
\begin{array}{cl}
\kappa_{fa}(\Gamma_\text{DD}) & \mbox{if } \tau>m_0(\Gamma_\text{DD}),\\
0 &  \mbox{if } \tau\leq m_0(\Gamma_{\text{DD}}),
\end{array}\right.
\end{equation}
\begin{equation}
\label{DDm}
\kappa_{m}^\text{DD}= 
\left\{
\begin{array}{cl}
\kappa_{m}(\Gamma_\text{DD}), & \mbox{if } \tau<m_1(\Gamma_\text{DD}),\\
0 &  \mbox{if } \tau\geq m_1(\Gamma_{\text{DD}}),
\end{array}
\right.
\end{equation}
where $t^*$ is the unique solution to (\ref{eq:t}) and $m_0(\cdot)$ and $m_1(\cdot)$ are given by (\ref{eq:m0}) and (\ref{eq:m1}) with $\Gamma=\Gamma_\text{DD}$, $\Theta=\Theta_\beta$.
\end{theorem}
\begin{proof}
See App. \ref{App:ProofDD}.
\end{proof}

\begin{corollary}[DD Miss Error Exponent and Optimum Energy Profile]
\label{cor:DDMiss}
Consider the hypothesis testing problem in (\ref{eq:testFC}). Assume that $\phi(\nu)$ is in the Wiener class, the fraction of DoF is limited to $\beta$, and the energy constraint is given by (\ref{xiConstr}). The optimal miss error exponent subject to $P_{fa}^n\leq \alpha$, with $\alpha\in(0,1)$ fixed, considering orthogonal schemes is achieved with either the PCS-MAC or PFS-MAC scheme and is given by
\begin{equation}
\label{eq:DDMiss}
\kappa_{m,\alpha}^\text{DD}\!=\!\!\int_{\Theta_{\beta*}} \!\!\left\{\frac{1}{1+\Gamma_\text{DD}(\nu)}\! + \log\left(1+\Gamma_\text{DD}(\nu)\right)-1 \right\}d\nu,
\end{equation}
where $\Theta_{\beta*}\subseteq \Theta_{\beta}$
is the support of the optimal energy profile $\xi_\text{OEP}(\nu)=\max\{\hat{\xi}_1(\nu),\hat{\xi}_2(\nu),\hat{\xi}_3(\nu),0\}$ and $\hat{\xi}_i(\nu)$, $i=1,2,3,$ are the roots of the following cubic equation:
\begin{equation}
a_3(\nu) \hat{\xi}^3+a_2(\nu)\hat{\xi}^2 + a_1(\nu) \hat{\xi}+a_0=0
\label{cubic}
\end{equation}
with coefficients
\begin{align}
\label{eq:cubicCoeff}
a_0\phantom{(\nu)}&=\lambda^*\sigma^6_w\nonumber\\
a_1(\nu)&=\sigma^2_w(-\phi(\nu)^2+\lambda^*\sigma^2_w(2\phi(\nu)+3\sigma^2_v))\\
a_2(\nu)&=\lambda^*\sigma^2_w (\phi(\nu)^2 + 4\phi(\nu)\sigma^2_v +3\sigma^4_v)\nonumber\\
a_3(\nu)&=\lambda^*\sigma^2_v(\phi(\nu)+ \sigma^2_v)^2,\nonumber
\end{align}
where $\lambda^*$ is the Lagrange multiplier that satisfies the energy constraint (\ref{xiConstr}) with equality. The closed-form solution of $\xi_\text{OEP}(\nu)$ is shown in (\ref{xiSol}), at the top of the next page.
\end{corollary}
\begin{proof}
See App. \ref{App:ProofMissExp} 
\end{proof}

\begin{figure*}[hbt]
\begin{align}
\label{xiSol}
\xi_\text{OEP}(\nu) &= \{(1/(62^{1/3} \lambda(\phi^2 v + 2\phi v^2 +  v^3)))((1 - \jmath \sqrt{3})
      (2 \lambda^3\phi^6 w^3 + 6 \lambda^3\phi^5 v w^3 + 6 \lambda^3\phi^4 v^2 w^3 + 2 \lambda^3\phi^3 v^3 w^3 +9 \lambda^2\phi^6 v w^2 \nonumber\\
      &+ 54 \lambda^2\phi^5 v^2 w^2 + 108 \lambda^2\phi^4 v^3 w^2 + 90 \lambda^2\phi^3 v^4 w^2 + 27 \lambda^2\phi^2 v^5 w^2 
        + (4(3 \lambda w(\phi^2 v + 2\phi v^2 +  v^3)(2 \lambda\phi w + 3 \lambda v w - \phi^2) \nonumber\\
      & -\lambda^2 w^2(\phi^2 + 4\phi v + 3 v^2)^2)^3 + (2 \lambda^3\phi^6 w^3 + 6 \lambda^3\phi^5 v w^3 + 6 \lambda^3\phi^4 v^2 w^3 + 2 \lambda^3\phi^3 v^3 w^3 + 9 \lambda^2\phi^6 v w^2 + 54 \lambda^2\phi^5 v^2 w^2\nonumber\\ 
      & + 108 \lambda^2\phi^4 v^3 w^2 + 90 \lambda^2\phi^3 v^4 w^2 + 27 \lambda^2\phi^2 v^5 w^2)^2)^{2/3}) - ((1 + \jmath \sqrt{3})(3 \lambda w(\phi^2 v + 2\phi v^2 +  v^3)(2 \lambda\phi w + 3 \lambda v w - \phi^2)\nonumber \\ 
      & -\lambda^2 w^2(\phi^2 + 4\phi v + 3 v^2)^2))/(32^(2/3) \lambda(\phi^2 v + 2\phi v^2 +  v^3)
      (2 \lambda^3\phi^6 w^3 + 6 \lambda^3\phi^5 v w^3 + 6 \lambda^3\phi^4 v^2 w^3 + 2 \lambda^3\phi^3 v^3 w^3\\
      & + 9 \lambda^2\phi^6 v w^2 + 54 \lambda^2\phi^5 v^2 w^2 + 108 \lambda^2\phi^4 v^3 w^2 + 90 \lambda^2\phi^3 v^4 w^2 + 27 \lambda^2\phi^2 v^5 w^2 + (4(3 \lambda w(\phi^2 v + 2\phi v^2 +  v^3)\nonumber\\ 
      & \times(2 \lambda\phi w + 3 \lambda v w - \phi^2) - \lambda^2 w^2(\phi^2 + 4\phi v + 3 v^2)^2)^3 + (2 \lambda^3\phi^6 w^3 + 6 \lambda^3\phi^5 v w^3 + 6 \lambda^3\phi^4 v^2 w^3 + 2 \lambda^3\phi^3 v^3 w^3\nonumber \\ 
      & + 9 \lambda^2\phi^6 v w^2 + 54 \lambda^2\phi^5 v^2 w^2 + 108 \lambda^2\phi^4 v^3 w^2 + 90 \lambda^2\phi^3 v^4 w^2 + 27 \lambda^2\phi^2 v^5 w^2)^2)^{2/3})- 
      ( w(\phi^2 + 4\phi v + 3 v^2))\nonumber\\
      &\;\,/(3(\phi^2 v + 2\phi v^2 +  v^3))\}\bigone(D(\nu)\geq 0). \nonumber\\ \nonumber\\   
D(\nu)&= -( \lambda \phi^2+2  \lambda \phi  v+ \lambda  v^2) (4  \lambda^2 \phi^5  w^5- \lambda \phi^6  w^4+18  \lambda \phi^5  v  w^4+27  \lambda \phi^4  v^2  w^4-4 \phi^6  v  w^3).
\label{discrim}
\end{align}
\caption{Optimal energy profile and the discriminant of (\ref{cubic}), $D(\nu)$. For readability, we define $\phi\defeq \phi(\nu)$, $w\defeq\sigma^2_w$, $v\defeq\sigma^2_v$ and $\lambda\defeq\lambda^*$. }
\end{figure*}

\subsection{Suboptimal Energy Profiles}
\label{sec:Profiles}
In this section we define three energy profiles: 
i) constant energy profile, $\xi_\text{CEP}(\nu) \defeq \frac{E_t}{(\sigma^2_v+\sigma^2_s)}$, where all the sensors transmit using the same gain for all channel uses;
ii) spectral energy profile, $\xi_\text{SEP}(\nu) \defeq \frac{E_t \phi(\nu)} {\sigma^2_s(\sigma^2_v+\sigma^2_s)}$, that reproduces the shape of $\phi(\nu)$ and therefore, allocates more energy to the frequencies where the process under $\Hip_1$ concentrates more power; iii) ON/OFF energy profile, $\xi_\text{ON/OFF-EP}(\nu) \defeq \frac{E_t}{(\sigma^2_v+\sigma^2_s)\beta^*} \bigone(\nu\in\Theta_{\beta*})$, with $\beta^*=\text{L}(\Theta_{\beta*})\leq \beta$,
where each sensor transmits with constant gain if $\nu\in\Theta_{\beta*}$ and stays silent otherwise. 
In i) and ii) we assume that the sensors do not have access to the transmitted modes set and no DoF compression is possible ($\beta=1$). In iii) and in the optimal energy profile, the sensors know the transmitted modes set. 

\section{Numerical Results for 1D}
\label{sec:results1D}

In this section, we evaluate the miss error exponent for two complex Gaussian correlated auto-regressive moving average (ARMA) processes with PSD given by 
$\phi(\nu)= \sigma^2_\text{in} \left|\frac{\sum_{k=0}^{M} b_k\, e^{-\jmath 2\pi \nu k}}{\sum_{k=0}^{N} a_k\, e^{-\jmath 2\pi \nu k}}\right|^2$, 
where $M$ and $N$ are the degrees of the numerator and denominator polynomials, respectively, and $\sigma^2_\text{in}$ is selected such that the variance of the process is $\sigma^2_s$. The simulated processes (with even PSD) are plotted in Fig. \ref{fig:PSDs} with coefficients shown in Table \ref{coef}. We will referred to them as PSD1 and PSD2.
\begin{figure}[bht]
\centering
\includegraphics[width=1\linewidth]{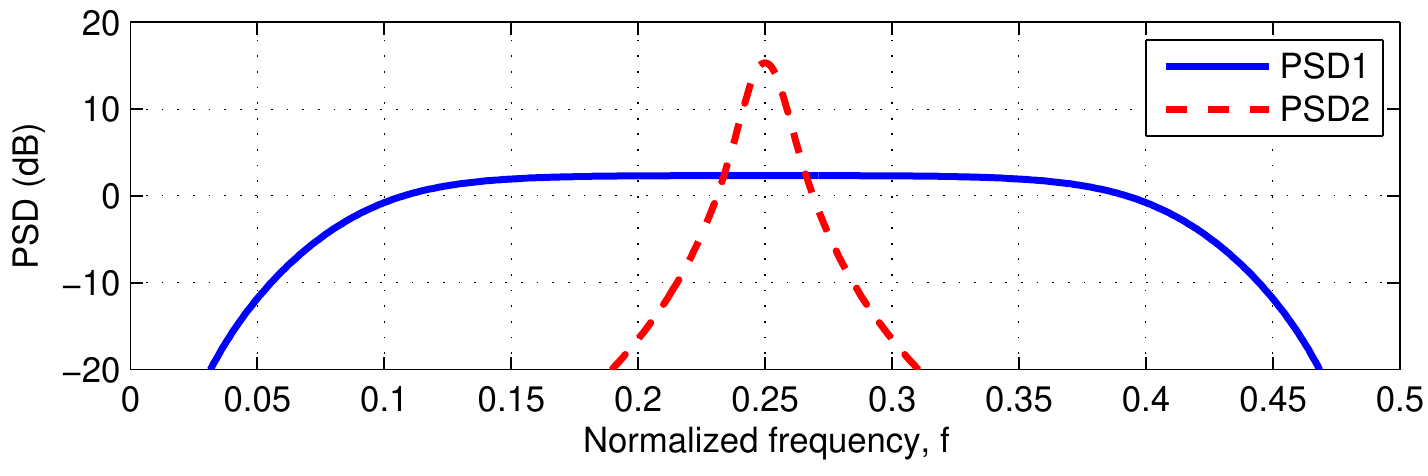}
\caption{PSDs (in dB) with variance $\sigma_s^2=1$. Only positive normalized frequencies are plotted.}
\label{fig:PSDs}
\end{figure}
\begin{table}[htb]
\begin{center}
\begin{tabular}{| c | c | c | c | c |}
\hline
\multicolumn{5}{|c|}{Coefficients PSD1:\phantom{\large k'}$\sigma^2_\text{in}=1.70$} \\
\hline\hline
$b_0$ & $b_1$ & $b_2$ & $b_3$ & $b_4$ \\
\hline
 $.39$ & $0$ & $-.78$& $0$& $.39$\\
\hline\hline
 $a_0$ & $a_1$ & $a_2$ & $a_3$ & $a_4$\\
\hline
 $1$ & $0$ & $-.37$ &  $0$ & $.19$\\
\hline
\end{tabular}
\begin{tabular}{| c | c | c | c | c |}
\hline
\multicolumn{5}{|c|}{Coefficients PSD2:\phantom{\large k'}$\sigma^2_\text{in}=2.37\cdot 10^{-5}$} \\
\hline\hline
$b_0$ & $b_1$ & $b_2$ & $b_3$ & $b_4$ \\
\hline
$3$ & $0$ & $-6$& $0$& $3$\\
\hline\hline
 $a_0$ & $a_1$ & $a_2$ & $a_3$ & $a_4$\\
\hline
$1$ &  $0$ &   $1.82$ &  $0$ &    $0.83$\\
\hline     
\end{tabular}
\end{center}
\caption{ARMA coefficients to generate both PSDs with variance $\sigma^2_s=1$ and $N=M=4$.}
\label{coef}
\end{table}
Let  $\text{SNR}_\text{M}=\sigma^2_s/\sigma^2_v$ and $\text{SNR}_\text{C}=E_t/\sigma^2_w$ 
be the measurement and communication signal-to-noise ratios, respectively. 
We use $\text{SNR}_\text{M}= 5$ dB, $\sigma^2_s=1$ and $E_t=1$ on all figures of this section.
In Fig. \ref{fig:ExpVsEta}, we show the behavior of the miss error exponent for processes PSD1 and PSD2 against $\text{SNR}_\text{C}$ when $P_{fa}^n\leq\alpha$.
In both cases the decentralized detectors (DD) approach the centralized detector (CD) performance when $\text{SNR}_\text{C}$ is high enough. We see that the optimum scheme DD-OEP allows to save a significant amount of energy when the communication signal-to-noise ratio is low. This is the desirable operation regime of a wireless sensor network with massive amount of nodes since it would extend the useful life of the WSN, avoid maintenance action for changing the battery of the nodes, or even elude network reconfiguration when nodes run out of energy. 

Table \ref{tab:energyGap} shows the energy savings of the optimum scheme DD-OEP and the suboptimal energy profile DD-SEP with respect to (wrt) the constant energy profile DD-CEP for several miss error exponents, and for both PSDs.
\begin{figure}[t]
\centering
\subfigure[PSD1, $\beta=0.6$ for OEP and ON/OFF. \label{fig:ExpVsEta1}]{\includegraphics[width=1\linewidth]{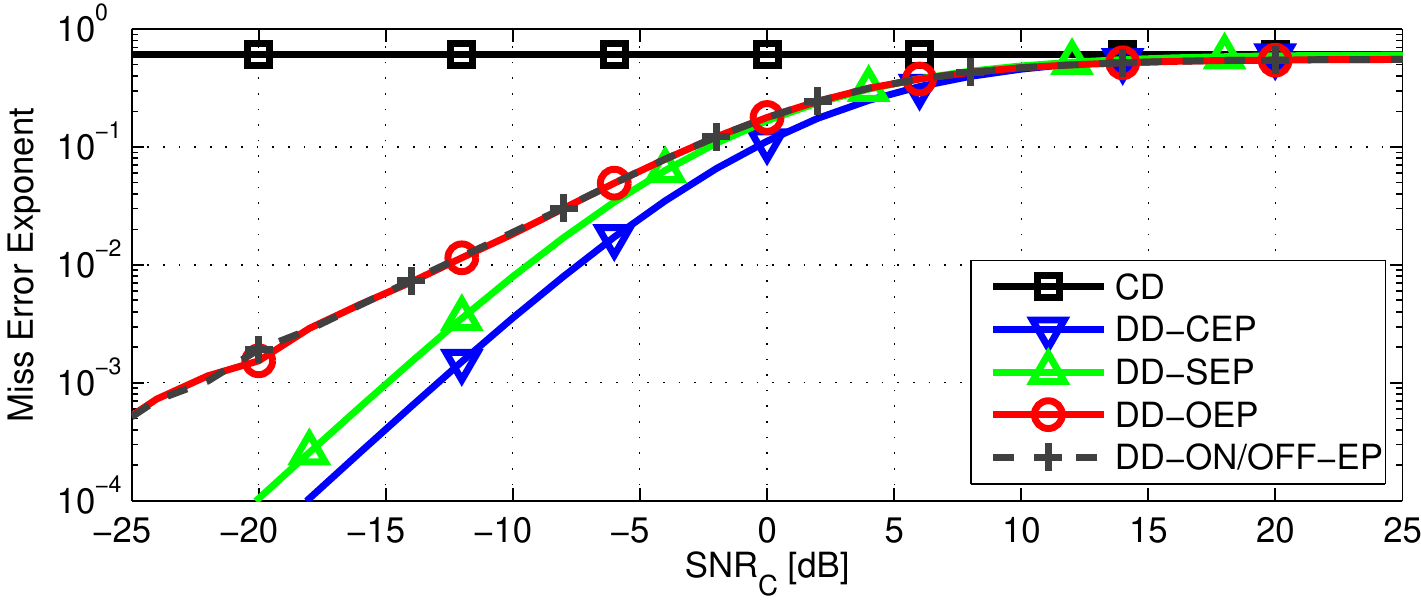}}
\subfigure[PSD2, $\beta=0.2$ for OEP and ON/OFF. \label{fig:ExpVsEta2}]{\includegraphics[width=1\linewidth]{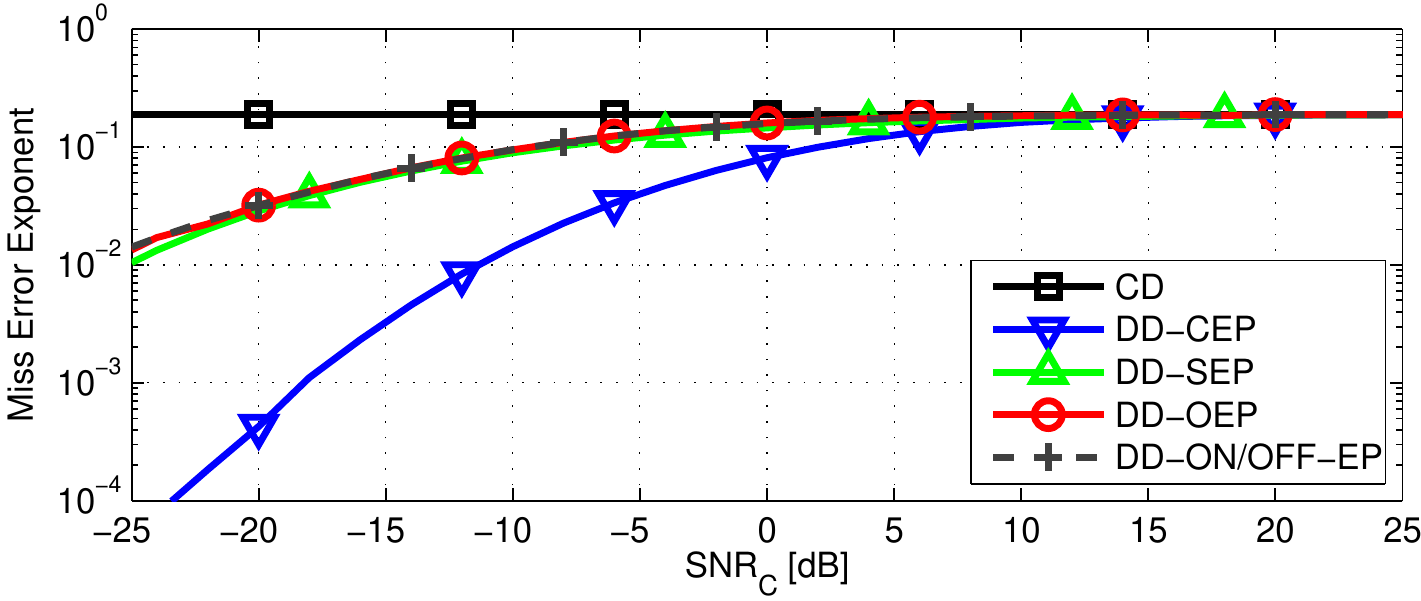}}
\caption{Error exponents for CD and DD detectors against $\text{SNR}_\text{C}$ (dB) for process PSD1 and PSD2.}
\label{fig:ExpVsEta}
\end{figure}
\begin{table}[b]
\begin{center}
\begin{tabular}{| c | c | c |}
\hline
\multicolumn{3}{|c|}{Energy gap (in dB) wrt CEP for PSD1} \\
\hline\hline
$\kappa_m$ & SEP & OEP, $\beta=0.6$\\
\hline
$10^{-1}$ & 2 & 2.5 \\
\hline
$10^{-2}$ & 2 & 5 \\
\hline
$10^{-3}$ & 2 & 10 \\
\hline
\hline
\end{tabular}
\begin{tabular}{| c | c | c |}
\hline
\multicolumn{3}{|c|}{Energy gap (in dB) wrt CEP for PSD2} \\
\hline\hline
$\kappa_m$ & SEP & OEP, $\beta=0.2$\\
\hline
$10^{-1}$ & 10.5 & 11  \\
\hline
 $10^{-2}$ & 14 & 15 \\
\hline
 $10^{-3}$ & >14 & >15 \\
\hline
\hline
\end{tabular}
\end{center}
\caption{Energy saving allocating energy in the sensors.}
\label{tab:energyGap}
\end{table}


A remarkable result is obtained when the scheme DD-ON/OFF-EP is used. Comparing with the optimal strategy, DD-OEP, we observe that the differences between the exponents obtained with both strategies are negligible. This shows that knowing $\Theta_{\beta*}$ is much more important than knowing the optimal gains $\{\gamma_{k'}^n\}$. 
This observation leads to a very simple strategy: at the beginning of the detection process, the sensors are communicated the optimal set $\Theta_{\beta*}$ (e.g., the FC broadcasts it through a low-rate feedback channel); then, on all channel uses, each node transmits with the same gain using PFS-MAC.

In Fig. \ref{fig:beta}, we plot the optimum fraction of DoF $\beta^*$ constrained to the allowed fraction of DoF $\beta$ against this constraint for processes PSD1 and PSD2. 
For low values of $\beta$, $\beta^*$ increases linearly with slope one ($\beta^*\leq \beta$) up to a certain value where $\beta^*$ saturates. The saturation effect is observed for different levels depending on $\text{SNR}_\text{C}$ and on the frequency selectivity of the process, which is related to its correlation. Thus, a strongly correlated process (PSD2) needs relatively less channel uses than a weakly correlated process (PSD1). When $\text{SNR}_\text{C}$ is high, saturation occurs for high values of $\beta$. Conversely, if $\text{SNR}_\text{C}$ is low, the optimum scheme uses the available energy $E_t$ in each sensor to transmit more reliably a reduced set of frequencies where the PSD of the process is high. 
As a remark, the asymptotically optimal orthogonal strategy PFS-MAC not only allows to save a valuable amount of energy in the sensors but also allows to save channel uses (i.e., bandwidth, detection delay).  
\begin{figure}[t]
\centering
\subfigure[PSD1.\label{fig:beta1}]{\includegraphics[width=1\linewidth]{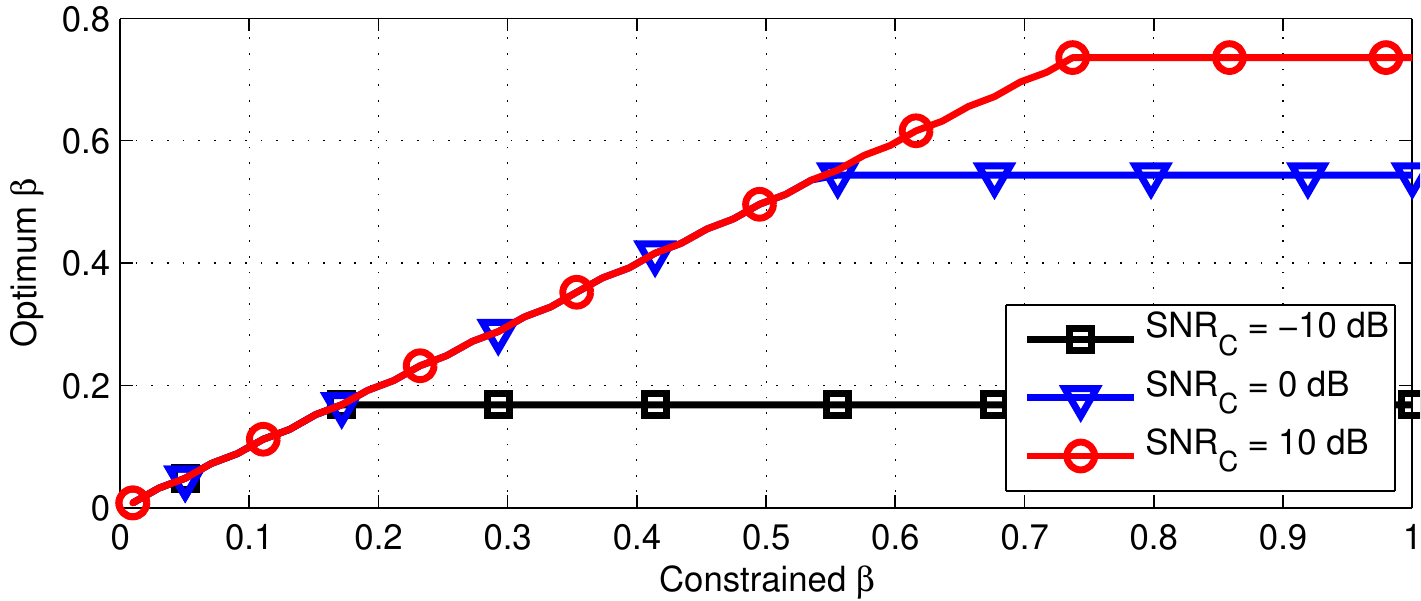}}
\subfigure[PSD2.\label{fig:beta2}]{\includegraphics[width=1\linewidth]{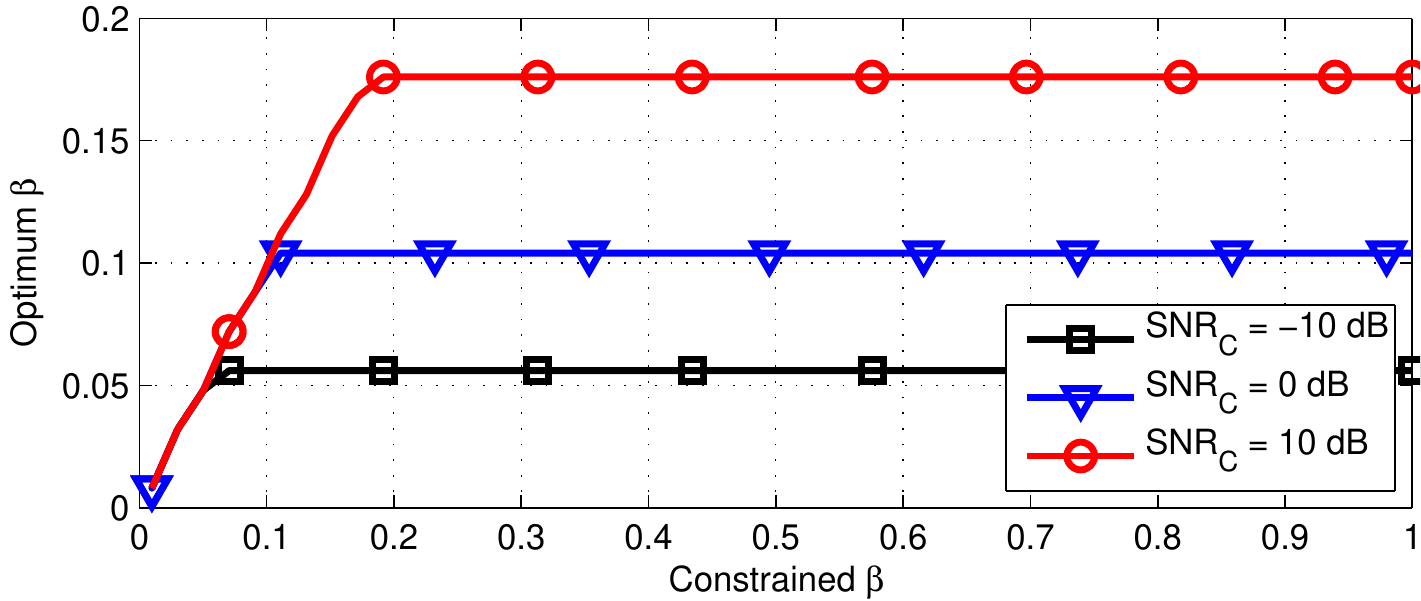}}
\caption{Optimum constrained fraction of DoF against the constrained fraction of DoF.}
\label{fig:beta}
\end{figure}
In Fig. \ref{fig:EP}, we plot the optimum energy profiles as a function of the normalized frequency for both processes and for several values of $\text{SNR}_\text{C}$, when the constraint on the fraction of DoF is inactive ($\beta=1$). Note that these figures are closely related to Fig. \ref{fig:beta} since the Lebesgue measure of the support of the optimum energy profile is indeed $\beta^*$. 
\begin{figure}[bt]
\centering
\subfigure[PSD1.\label{fig:EP1}]{\includegraphics[width=1\linewidth]{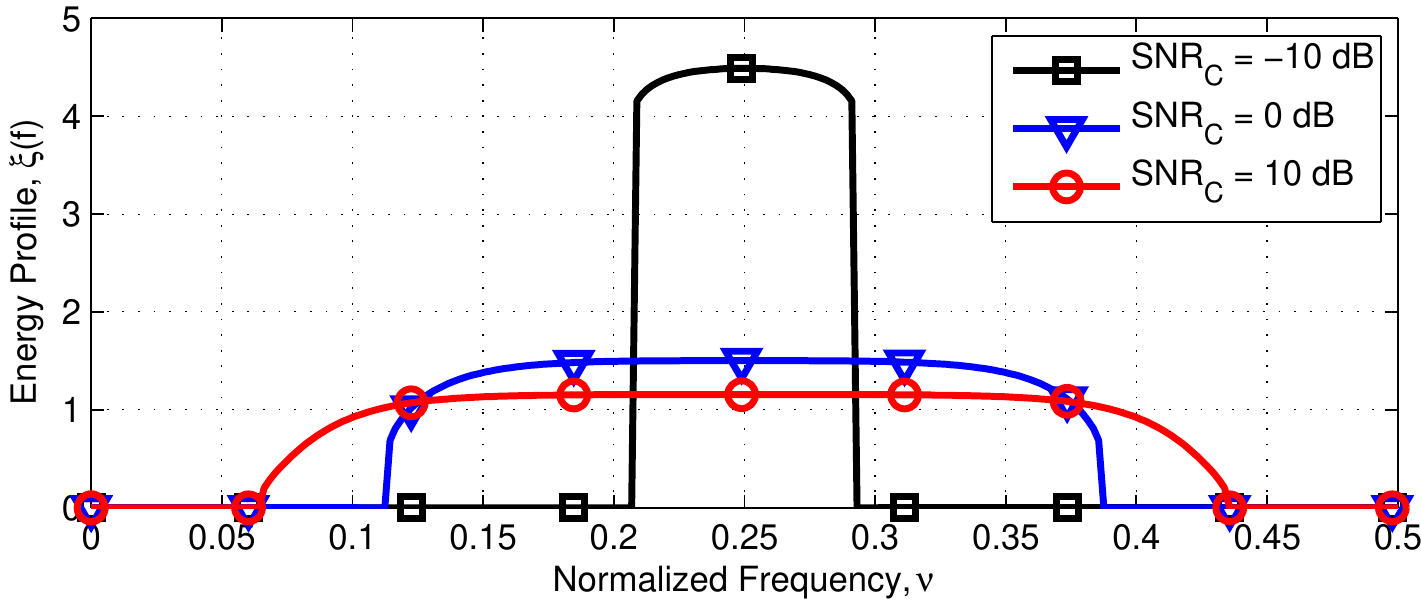}}
\subfigure[PSD2.\label{fig:EP2}]{\includegraphics[width=1\linewidth]{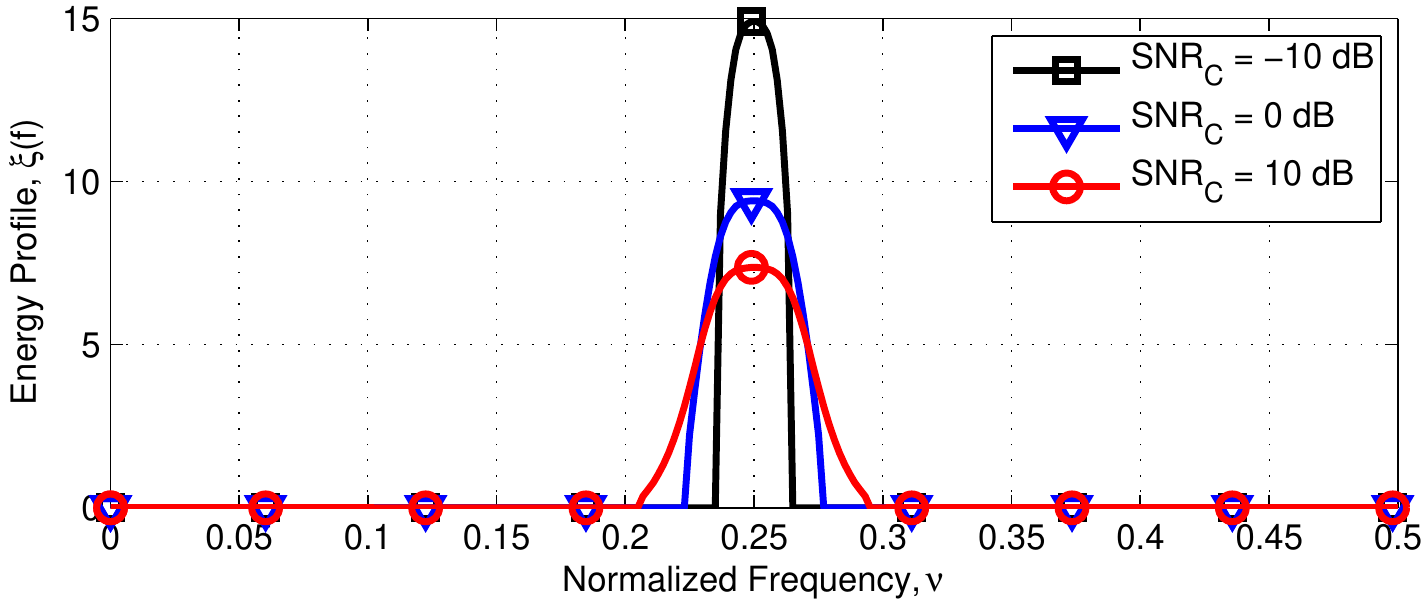}}
\caption{Energy profiles for DD-OEP detector against the normalized frequency for process with PSD1 and PSD2 without DoF constraint ($\beta=1$). Only $[0,0.5]$ frequency interval is shown due to the even symmetry of $\xi(\nu)$ inherited from $\phi(\nu)$.}
\label{fig:EP}
\end{figure}

\section{2 and 3 Dimensional WSNs}
\label{sec:p-WSN}
In this section we enlarge the model to networks with up to 2 dimensions in space and we include the time dimension. The extension to more dimensions is straightforward. 
The measurement taken by the sensor at coordinate $(i,j)$ at time $m$ under each hypothesis is:
\begin{equation*}
\left\{
\begin{array}{lllcl}
\Hip_1: & x_{ij}(m)&=s_{ij}(m) + v_{ij}(m),& i&=1,\dots,n_1,\\
\Hip_0: & x_{ij}(m)&=v_{ij}(m),& j&=1,\dots,n_2,\\
& & &   m&=1,\dots,n_3.
\end{array}
\right.
\end{equation*}
where the parameters $n_1$, $n_2$, and $n_3$ are chosen for describing one, two, or three dimensional networks. In these networks, sensors are distributed along a line (1D) or in space (2D), and they take several measurements during different time instants (time dimension). We assume that $s_{ij}(m)$ is a zero-mean circularly-symmetric complex Gaussian stationary process with variance $\sigma^2_s$ and power spectral density (PSD) $\phi(\ve{\nu})$, where $\ve{\nu}\in[0,1]^p$ is the vector of frequencies normalized to the interval $[0,1]$ and $p$ is the dimension of the network. 
$v_{ij}(m)$ is zero-mean circularly-symmetric complex white Gaussian noise independent of $s_{ij}(m)$ with variance $\sigma_v^2$. Therefore, $x_{ij}(m)$ is Gaussian distributed under each hypothesis. 

All sensors transmit synchronously over a MAC, and the received signal at the FC during the $k'$-th channel use is: 
\begin{equation}
z_{k'}^n=\sum_{i=1}^{n_1} \sum_{j=1}^{n_2} g_{ij k'}(\{x_{ij}(m)\}_{m=1}^{n_3}) + w_{k'},\ \ k'=1,\dots,n',
\label{eq:zpD}
\end{equation}
where $g_{ij k'}(\{x_{ij}(m)\}_{m=1}^{n_3})$ is the symbol transmitted by the sensor located at $(i,j)$ during channel use $k'$ using the encoding function $g_{ij k'}(\cdot)$, $w_{k'}$ is a zero-mean circularly-symmetric complex white Gaussian noise independent of everything else with variance $\sigma_w^2$, and $n'$ is the number of MAC uses. As in the one dimensional case, we will consider linear encoding functions only. To each 3-tuple index $(i,j,m)$, we associate the following unwound index 
\begin{equation}
k=(i-1)n_2 n_3 + (j-1)n_3 + m
\label{map}
\end{equation} 
that takes values $k=1,\dots,n$ with $n=n_1 n_2 n_3$. This is a one to one mapping, i.e., each 3-tuple index $(i,j,m)$ can be recovered from the index $k$ and vice versa. We denote this as $k\leftrightarrow (i,j,m)$. If we define the following column vectors: $\ve{s}=[s^n_1,\dots,s^n_n]^T$, $\ve{v}=[v^n_1,\dots, v^n_n]^T$ and $\ve{x}=[x_1^n,\dots, v^n_n]^T$, the vector of measurements at the FC are expressed as in (\ref{eq:zvec}) and hence, its statistic  results (\ref{TD1}). However, we assume now that the covariance matrix $\Sigma_n$ is a $p$-level Toeplitz matrix \cite[Sec. 6.4]{Tilli:Toeplitz} instead of a regular Toeplitz matrix.
The precoding matrix for the $p$-dimensional ($p$-D) PCS-MAC strategy has the same expression (\ref{PCS-matrix}) although the eigenvectors of the $p$-level Toeplitz matrix $\Sigma_n$ change. However, we need to redefine the precoding matrix for the $p$-D PFS-MAC strategy.

\begin{definition}[$p$-D PFS-MAC]
\label{PFS-MAC-pD}
For each $n=n_1\times\cdots\times n_p$, denote by $(j_1,j_2,\dots,j_n)$ a permutation of $\{1,2,\dots,n\}$ such that $\phi[j_1]\geq\phi[j_2]\geq \cdots \geq \phi[j_n]$ where $j_{k'}\leftrightarrow(i_{1}',\dots,i_{p}')$, $i_{l}'\in [1,n_l]$, $l=1,\dots,p$ and $\phi[j_{k'}]$ is defined as a sample of the $p$-D PSD: $\phi[j_{k'}]\defeq\phi\left(\frac{i_{1}'-1}{n_1},\dots,\frac{i_{p}'-1}{n_p}\right)$ with $k'=1,\dots,n'$. The precoding matrix of the $p$-D PFS-MAC strategy is
\begin{equation}
C_{nn'}=F_{nn'}\Delta_{n'},
\label{PFS-matrix-pD}
\end{equation}
where $F_{nn'}=[\ve{f}_{j_1^n}^n,\dots,\ve{f}_{j_{n'}^n}^n]$ is a sub-matrix of the $p$-D DFT matrix of size $n\times n$, i.e.,  
$\ve{f}^n_{k'}=[f_{1k'}^n,f_{2k'}^n,\dots, f_{nk'}^n]^T$ with $f_{kk'}^n = \frac{1}{\sqrt{n}}\exp(\jmath 2\pi \sum_{l=1}^p \frac{(i_l-1) (i_{l,k'}'-1)}{n_l})$, $k=1,\dots,n$; $k\leftrightarrow(i_1,\dots,i_p)$, and $\Delta_{n'}=\diag(\gamma^n_{1},\dots,\gamma^n_{n'})$.
\end{definition}

Similar to the 1D case, the objective of both $p$-D PCS-MAC and $p$-D PFS-MAC is to communicate the most important modes (frequencies) of the random process to the FC. Under both strategies, each sensor needs to know the complete vector of time measurements to transmit during each channel use a different $p$-tuple index $(i_{1,k'}', \dots,i_{p,k'}')$. Once the FC has the $n'$ measurements, it builds the statistic (\ref{TD1}) to make a decision.
 
 We need now a multi-dimensional version of the Toeplitz theorem.
\begin{theorem}[Toeplitz Distribution for $p$-D processes]
\label{ToepTheo-pD}
For a Hermitian $p$-level Toeplitz matrix $\Sigma_{n}$ generated by the spectral density $\phi(\ve{\nu})$ which belongs to the Wiener class, with $\ve{\nu}=(\nu_1,\dots,\nu_p)$ and multilevel index $\ve{n}=(n_1,\dots,n_p)$, $n=n_1\times\cdots\times n_p$, let $\{\lambda_{k}^n\}_{k=1}^{n}$ be the eigenvalues of $\Sigma_{n}$ contained on the interval $[\delta_1 , \delta_2]$, let $F(\cdot)$ be a continuous function defined on $[\delta_1 , \delta_2]$ and assume that $\int_{\ve{\nu}:\phi(\ve{\nu})=\delta}F(\phi(\ve{\nu}))d\ve{\nu} =0$ 
for any $\delta\in [\delta_1 , \delta_2]\cap \Delta$ where $\Delta\subseteq [0,\infty)$ then
\begin{equation*}
\lim_{\ve{n}\rightarrow\infty}\
\frac{1}{n}\sum_{k=1}^{n} F(\lambda_{k}^n\in \Delta ) = \int_{\phi^{-1}(\Delta)} F(\phi(\ve{\nu}))d\ve{\nu}.
\end{equation*}
where 
$\ve{n}\rightarrow\infty$ means that all components of $\ve{n}$ tend to infinity simultaneously.
\end{theorem}
\emph{Proof:} See \cite[Th. 6.4.1]{Tilli:Toeplitz} and \cite[Corollary 4.1]{GrayToeplitz}.

Using Th. \ref{ToepTheo-pD} and Th. \ref{GETheo} together with the $p$-D PCS-MAC and $p$-D PFS-MAC strategies we obtain  the same results given in Section \ref{sec:exponents}. This is summarized in the following theorem.
\begin{theorem}[DD Error Exponents for $p$-D networks]
\label{DDexp-pD}
The error exponents for the $p$-D PCS-MAC and $p$-D PFS-MAC strategies are given by Th. \ref{DDexp} and Corollary \ref{cor:DDMiss} under the same hypotheses considering now that the normalized frequency is a $p$-dimensional variable $\ve{\nu}$ and $\Theta$ is a $p$-dimensional set.
\end{theorem}
\begin{proof}
Apply Th. \ref{DDexp} using (\ref{map}). 
\end{proof}

The same energy profiles of subsection \ref{sec:Profiles} are defined if we consider again that the normalized frequency is a $p$-dimensional variable $\ve{\nu}$ and $\Theta$ is a $p$-dimensional set.

\section{Numerical experiment for 2D}
\label{sec:Applications}
In this section we compare the theoretical results with a Monte Carlo simulation for detecting a 2D correlated process described by the following partial differential equation:
\begin{align}
\tilde{a}_x \frac{\partial^2 s(x,y)}{\partial x^2} + \tilde{a}_y\frac{\partial^2 s(x,y)}{\partial y^2} + \tilde{a}_0 s(x,y) = q(x,y)
\label{pde}
\end{align} 
where $q(x,y)$ is the random source, assumed to be white and Gaussian with PSD $\sigma^2_q$, 
and $\tilde{a}_x$, $\tilde{a}_y$ and $\tilde{a}_0$ are constants. Using a second order approximation, we discretize (\ref{pde}),
\begin{align}
\frac{\partial^2 s(x,y)}{\partial x^2} &\approx\frac{s(x+h,y)-2s(x,y)+s(x-h,y)}{h^2}\\
\frac{\partial^2 s(x,y)}{\partial y^2} &\approx\frac{s(x,y+h)-2s(x,y)+s(x,y-h)}{h^2}
\end{align}
 where $h$ is the discretization step in both directions. The discrete equation results
 \begin{align}
 a_x (s_{i+1,j}+s_{i-1,j}) + a_y (s_{i,j+1} + s_{i,j-1}) + a_0 s_{i,j} = q_{i,j}, \nonumber\\
 i=1,\dots,n_1,\; j=1,\dots,n_2.\nonumber
 \label{fde}
 \end{align} 
 with $a_x=\tilde{a}_x/h^2$, $a_y=\tilde{a}_y/h^2$ and $a_0=\tilde{a}_0- 2(\tilde{a}_x + \tilde{a}_y)/h^2$. The PSD of the process is 
 \begin{align}
 \phi(\omega_x,\omega_y)= \frac{\sigma^2_q}{(a_0 + 2 a_x\cos(\omega_x) + 2 a_y\cos(\omega_y))^2}
 \end{align}
and its variance is $\sigma_s^2$. If we organize the samples of the process in a vector as in Section \ref{sec:p-WSN}, (\ref{fde}) results a linear system: $A_n \ve{s}_n = \ve{q}_n$. In there, $\ve{s}_n$ and $\ve{q}_n$ are vectors of dimension $n=n_1 n_2$, and $A_n$ is a two-level banded Toeplitz matrix of size $n\times n$ with the sub-block matrix $B_{n_1}$:
\begin{equation*}
A_n=\begin{bmatrix}
B_{n_1} &a_x I_{n_1}  	  \\ 
a_x I_{n_1}& B_{n_1} & \ddots   \\ 
& \ddots  &\ddots \\ 
\end{bmatrix},\,
B_{n_1}=\begin{bmatrix}
a_0 & a_y  	  \\ 
a_y & a_0 & \ddots   \\ 
 &\ddots  &\ddots\\ 
\end{bmatrix},
\end{equation*}
where only non-zero elements are indicated. Assuming that $A_n$ is non-singular, the covariance matrix of the process is
\begin{equation}
\label{eq:cov2D}
\Sigma_n=\sigma^2_q (A_n A_n^T)^{-1}
\end{equation}
which is in general a non-Toeplitz matrix. However, $A_n$ is asymptotically a circulant matrix. The product of two circulant matrices, as well as the inverse of a circulant matrix, are circulant matrices \cite[p. 50, 63 and 67]{GrayToeplitz}.
Circulant matrices are a particular case of Toeplitz matrices. Therefore, $\Sigma_n$ is asymptotically a 2-level Toeplitz matrix and we apply Th. \ref{DDexp-pD} to compute the miss error exponent subject to any fixed level of false alarm error probability $\alpha\in(0,1)$.

In Fig. \ref{fig:Mont} we show the estimation of the  miss error exponent as a function of the number of sensors in the network using the Monte Carlo method with the following parameters: $\mathrm{SNR}_\mathrm{C}=-10$ dB, $a_0=-5$, $a_x = a_y=1$, $10^6$ experiments and $\alpha=10^{-2}$. 
We consider that the sensors are placed in a regular square grid, i.e., $n_1=n_2$ and the total amount of sensors in the network is $n=n_1 n_2$. The threshold of the test is modified for each $n$ in order to keep the false alarm error probability constant. We also plot the theoretical miss error exponent. We observe that both estimated miss error exponents converge to their corresponding theoretical values. We also note that both decentralized schemes PCS-MAC and PFS-MAC have almost the same performance for the number of sensor considered in the figure, which validates the circulant approximation of the product of Toeplitz matrices in (\ref{eq:cov2D}).       
\begin{figure}[htb]
\centering
\includegraphics[width=\linewidth]{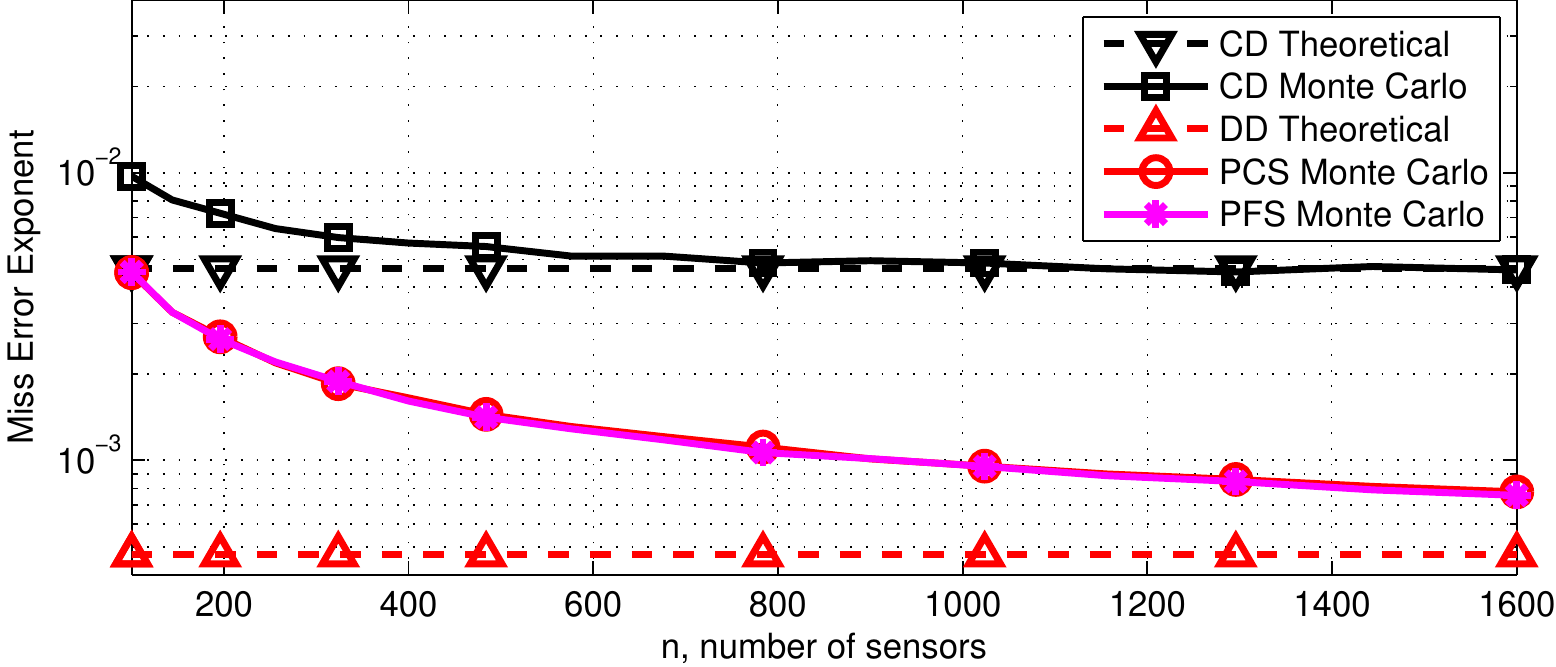}
\caption{Miss error probability vs. the amount of sensors in the network.}
\label{fig:Mont}
\end{figure}


\section{Conclusions}
\label{sec:Conclusions}
We have proposed several schemes for distributed detection of 
circularly-symmetric complex Gaussian random processes with arbitrary correlation function both in space and time. These schemes take into account possible correlated measurements and use this correlation beneficially at the FC to build an appropriate statistic to make a decision. Considering a multiple access channel and imposing bandwidth and per-sensor energy constraints, we have obtained the optimal orthogonal 
scheme in terms of the miss error exponent.   
We have also shown that one of the proposed schemes is particularly attractive for WSNs with low-cost and energy-limited nodes because it requires only the set of modes to be transmitted, obtains significant energy saving in the low signal-to-noise ratio regime, 
and achieves a close-to-optimal (with negligible loss) performance in terms of miss error exponent. 

\appendices

\section{Proof of Theorem \ref{CDexp}: CD Error Exponents}
\label{App:ProofCD}
The asymptotic mean of the centralized LLR statistic defined in (\ref{TC1}) under the hypothesis $\Hip_i$, $i=0,1$ is
\begin{align*}
m^\text{CD}_i &=\lim_{n\rightarrow\infty} \Ex_i(T_c^n)\nonumber\\
&= \lim_{n\rightarrow\infty} \frac{1}{n}\left(\tr\left( \Sigma_{i,n}\left(\Sigma_{0,n}^{-1}-\Sigma_{1,n}^{-1} \right)\right) -  \log\frac{\det\Sigma_{1,n}}{\det\Sigma_{0,n}}\right).
\end{align*}
Considering that the eigenvalues of $\Sigma_{0,n}$ and $\Sigma_{1,n}$ are $\{\sigma^2_v\}$ and $\{\lambda_k^n+\sigma^2_v\}$, respectively, and defining $\rho^n_{\text{CD},k}=\frac{\lambda_k^n}{\sigma_v^2}$, we obtain 
$m^\text{CD}_0 
=\lim_{n\rightarrow\infty}\frac{1}{n} \sum_{k=1}^n 1-\frac{1}{1+\rho_{\text{CD},k}^n} -\log(1+\rho_{\text{CD},k}^n)$.
Using Th. \ref{ToepTheo} with $\Delta=[\delta_1,\delta_2]$, which implies $\phi^{-1}(\Delta)=[0,1]$, we finally have $m^\text{CD}_0 = m_0(\Gamma_\text{CD})$ defined (\ref{eq:m0}). 

To compute the error exponents we first need to verify the assumptions of Th. \ref{GETheo} with $h_0\defeq h_0(\nu)=\sigma^2_v$ $\forall \nu$ and $h_1\defeq h_1(\nu)= \phi(\nu)+ \sigma^2_v$: 
\begin{itemize}
\item[$A_1$)] $h_0$ is a positive constant and $\log h_0$ belongs to $L^1([0,1])$ trivially. $\log h_1\in L^1([0,1])$ is easily proved by noting that $\ess\inf (\phi(\nu)+\sigma^2_v)>0$ and $\ess\sup (\phi(\nu)+\sigma^2_v)<\infty$
because $\phi(\nu)$ is a power spectral density in the Wiener class, and therefore, it is essentially bounded.

\item[$A_2$)] $h_0/h_1= 1/(1+\Gamma_\text{CD})\in L^\infty([0,1])$ given that $\ess\inf\Gamma_\text{CD}\geq 0$ because $\Gamma_\text{CD}$ is a spectral density. $h_1/h_0= 1+\Gamma_\text{CD} \in L^\infty([0,1])$ is easily proved by considering again that $\phi(\nu)$ belongs to the Wiener class.
\end{itemize} 
Given that $(A_1)$ and $(A_2)$ are satisfied, the error exponents are obtained through the Fenchel-Legendre transform $\Lambda_{\text{CD},i}^*(x)$ of the LMGF $\Lambda_{\text{CD},i}(t)=\log\Ex_i\left[e^{t T_c^n}\right]$, $i=0,1$.
The following properties can  be verified: 
\begin{properties}[LMGF and its Fenchel Legendre Transform]
\leavevmode
\begin{itemize}
\item[$P_1)$] $\Lambda_{\text{CD},1}(-1)=\Lambda_{\text{CD},1}(0)=0$. 
\item[$P_2)$] $\Lambda_{\text{CD},1}'(-1)=m^\text{CD}_0$ and $\Lambda_{\text{CD},1}'(0)=m^\text{CD}_1$. 
\item[$P_3)$] $\Lambda_{\text{CD},1}(t)$ is a convex function and then $\Lambda_{\text{CD},1}(t)=\Lambda_{\text{CD},1}((1+t)0 - t(-1)) \leq (1+t)\Lambda_{\text{CD},1}(0) -t\Lambda_{\text{CD},1}(-1)=0$ $\forall t\in[-1,0]$. 
\item[$P_4)$] As the Fenchel-Legendre transform given in (\ref{FLT:Hi}) transforms convex functions into convex functions, $\Lambda^*_{\text{CD},1}(x)$ is a convex function. Moreover, $\Lambda^*_{\text{CD},1}(x)$ has a minimum of value 0 at $x=m^\text{CD}_1$ meaning that $\Lambda^*_{\text{CD},1}(x)$ is decreasing if $x<m^\text{CD}_1$.
\item[$P_5)$] $\Lambda_{\text{CD},0}(t)=\Lambda_{\text{CD},1}(t-1)$ and $\Lambda_{\text{CD},0}^*(x)=\Lambda_{\text{CD},1}^*(x) +x$. Therefore, $\Lambda_{\text{CD},0}^*(x)$ is also a convex function with a minimum of value 0 at $x=m^\text{CD}_0$ and it is an increasing function if $x>m^\text{CD}_0$.
\end{itemize}
\label{prop:LGMF}
\end{properties}

Using properties $(P_1)$-$(P_5)$ and (\ref{LDT4}) when $G=[\tau,\infty)$, $k_{fa}^\text{CD}=\Lambda_{\text{CD},0}^*(\tau)$ if $\tau>m_0^\text{CD}$ and $k_{fa}^\text{CD}=0$ if $\tau\geq m_0^\text{CD}$. Moreover, when $G=(-\infty,\tau]$, $k_{m}^\text{CD}=\Lambda_{\text{CD},1}^*(\tau)$ if $\tau < m_1^\text{CD}$ and $k_{m}^\text{CD}=0$ if $\tau\geq m_1^\text{CD}$. The usual interval of interest for the threshold $\tau$ is $[m_0^\text{CD},m_1^\text{CD}]$ and because $(P_1)$ and $(P_2)$ hold, the interval of optimization of $t$ in (\ref{FLT:Hi}) can be restricted to $[-1,0]$. Due to the convexity of $\Lambda_{\text{CD},1}(t)$ there exists a unique $t^*$ that solves (\ref{FLT:Hi}) and satisfies (\ref{eq:t}), which is obtained from deriving $\tau t-\Lambda_{\text{CD},1}(t)$. Then, the error exponents are (\ref{CDfa}) and (\ref{CDm}).

\section{Decentralized Error Exponents}
\subsection{Proof of Lemma \ref{lem:HT-DD}: Decentralized Hypothesis Testing}
\label{App:HT-DD}
We first consider the PCS-MAC scheme given in Def. \ref{PCS-MAC}. The covariance matrices under $\Hip_0$ and $\Hip_1$ are, respectively,
\begin{align}
\Xi_{0,n'}&=\sigma^2_v \Delta^2_{n'} + \sigma^2_w I_{n'},\nonumber\\
\Xi_{1,n'}&= ( \diag(\lambda_1^n,\dots,\lambda_{n'}^n)+ \sigma^2_v I_{n'}) \Delta^2_{n'} + \sigma^2_w I_{n'}.
\label{eq:PCSCov1}
\end{align}
Let $\theta_{i,k}$ be the $k$-th element of the diagonal matrix $\Xi_{i,n'}$, $i=0,1$. It is easy to prove that $\diag(\theta_{i,1},\dots,\theta_{i,n'})\sim \diag(h_{i}(\frac{0}{n}),\dots,h_{i}(\frac{n'-1}{n}))$. 
To recover the original dimension of the problem, define the $n$-dimensional vector $\tilde{\ve{z}}$ as the zero padding of the $n'$-dimensional vector of measurements $\ve{z}$, i.e., $\tilde{\ve{z}}=[\ve{z}^T \ve{0}^T]^T$. The covariance matrix of $\tilde{\ve{z}}$, $\diag(\theta_{i,1},\dots,\theta_{i,n'},0,\dots,0)$, is asymptotically equivalent to $\diag(h_{i}(\frac{0}{n}),\dots,h_{i}(\frac{n'-1}{n}),0,\dots,0)$.
Consider the following transformation $\ve{y}=F_n \tilde{\ve{z}}$, where $F_n$ is the DFT matrix of size $n\times n$. 
It is well known that applying an invertible transformation (in particular, the  orthogonal DFT matrix) to the data does not modify the performance of the statistic. Because matrix multiplication preserves asymptotic equivalence of matrices \cite[Th. 2.1 (3)]{GrayToeplitz}, we have that the covariance matrix of $\ve{y}$ is asymptotically circulant, i.e.,
\begin{align*}
& F_n\diag(\theta_{i,1},\dots,\theta_{i,n'},0,\dots,0)F_n^H\sim\\ 
& F_n\diag\left(h_{i}(0/n),\dots,h_{i}((n'\!-1)/n),0,\dots,0\right)F_n^H\! = B_n(\tilde{h}_i),
\end{align*}
where $B_n(\tilde{h}_i)$ is the circulant matrix generated by the samples of $\tilde{h}_i(\nu)$, $i=0,1$ and
\begin{equation}
\tilde{h}_i(\nu)=\left\{
\begin{array}{lll}
h_i(\nu) & \text{if} & \nu\in\Theta_\beta\\
0 & \text{if} & \nu\in[0,1]\setminus \Theta_\beta.
\end{array}
\right.
\end{equation} 
Therefore, $\tilde{h}_i(\nu)$ may be interpreted as the PSD of the measurements available at the FC under $\Hip_i$, $i=0,1$ and the test (\ref{eq:testFC}) can be formulated.

In the case of PFS-MAC, the covariance matrices under $\Hip_0$ and $\Hip_1$ are, respectively,
\begin{align}
\Xi_{0,n'}&=\sigma^2_v \Delta^2_{n'} + \sigma^2_w I_{n'},\nonumber\\
\Xi_{1,n'}&= \Delta_{n'}( F_{nn'}^H\Sigma_n F_{nn'} + \sigma^2_v I_{n'}) \Delta_{n'} + \sigma^2_w I_{n'}.
\label{eq:PFSCov1}
\end{align}
In this case, $F_{nn'}^H\Sigma_n F_{nn'}$ is not diagonal for finite $n'$. However, considering that asymptotic equivalence between matrices is preserved by matrix multiplication \cite[Th. 2.1 (3)]{GrayToeplitz}, and using Lem. \ref{lem:equiv} we have that $\Sigma_n\sim F_{nn'}\diag(\lambda_1^n,\dots,\lambda_{n'}^n) F_{nn'}^H $ if and only if 
\begin{equation}
F_{nn'}^H\Sigma_n F_{nn'}\sim \diag(\lambda_1^n,\dots,\lambda_{n'}^n),
\label{eq:asympEq}
\end{equation}
which makes (\ref{eq:PCSCov1}) and (\ref{eq:PFSCov1}) asymptotically equivalent, obtaining the same test of hypothesis (\ref{eq:testFC}).

\subsection{Proof of Theorem \ref{DDexp}: DD Error Exponents}
\label{App:ProofDD}
The asymptotic mean of the decentralized LLR statistic defined in (\ref{TD1}) under $\Hip_i$, $m^\text{DD}_i =\lim_{n\rightarrow\infty} \Ex_i(T_d^n)$, $i=0,1,$ depends entirely on the spectral density under $\Hip_i$.  A similar situation  occurs for the centralized detector in Th. \ref{CDexp}. Then, using Lem. \ref{lem:HT-DD}, we have that $m^\text{DD}_i=m_i(\Gamma_\text{DD})$ for both strategies PCS-MAC and PFS-MAC, where $m_i(\cdot)$, $i=0,1$, are defined in (\ref{eq:m0}) and (\ref{eq:m1}), respectively.

To compute the error exponents we first need to verify the assumptions of the modified version of the Gärdner-Ellis theorem with $h_0\defeq h_0(\nu)= \sigma_v^2\xi(\nu) + \sigma^2_w$ and $h_1\defeq h_1(\nu)=(\phi(\nu)+\sigma_v^2)\xi(\nu) + \sigma^2_w$:
\begin{itemize}
\item[$A_1)$] $\int_{\Theta_\beta}|\log(\xi(\nu)\sigma^2_v+ \sigma^2_w)|d\nu\leq\int_{\Theta_\beta}|\log(1+\frac{\xi(\nu)\sigma^2_v}{\sigma^2_w})|d\nu + |\log(\sigma^2_w)|\leq 1+ \frac{\sigma^2_v}{\sigma^2_w}\int_{\Theta_\beta}\xi(\nu)d\nu + |\log(\sigma^2_w)|<\infty$ 
because $\xi(\nu) \geq 0$, the energy constraint by (\ref{xiConstr}) and $|\log(x)|< |x|$ if $x\geq 1$. Then, $\log h_0\in L^1([0,1])$. $h_1\in L^1([0,1])$ is proved similarly by additionally  considering that $\phi(\nu)$ is a power spectral density in the Wiener class. 

\item[$A_2)$] $h_0/h_1= 1/(1+\Gamma_\text{DD})\in L^\infty([0,1])$ given that $\ess\inf\Gamma_\text{DD}\geq 0$ since $\xi(\nu)\geq 0$ and $\ess\inf \phi(\nu)\geq 0$. $h_1/h_0= 1+\Gamma_\text{DD} \in L^\infty([0,1])$ is easily proved by considering again that $\phi(\nu)$ belongs to the Wiener class and therefore $\ess\sup\phi(\nu)=M_\phi$ is a finite constant. Then, $\ess\sup\Gamma_\text{DD}\leq M_\phi/\sigma^2_v$.
\end{itemize}
Now, we apply Th. \ref{GETheo} to obtain the error exponents by computing the Fenchel-Legendre transforms $\Lambda_{\text{DD},i}^*(x)$ of the LMGF $\Lambda_{\text{DD},i}(t)=\log\Ex_i\left[e^{t T_d^n}\right]$, $i=0,1$. 
The same properties (P1)-(P5) in Prop. \ref{prop:LGMF} are satisfied by $\Lambda_{\text{DD},i}(t)$ and $\Lambda^*_{\text{DD},i}(x)$ considering $m^\text{DD}_i$ instead of $m^\text{CD}_i$, for $i=0,1$. Therefore, the error exponents are (\ref{DDfa}) and (\ref{DDm}).

\subsection{DD Optimum Energy Profile for the Miss Error Exponent}
\label{App:ProofMissExp}
If the false alarm probability constraint is $P_{fa}^n\leq\alpha<1$, the miss error exponent is given by Th. \ref{DDexp} with $\tau=m_0(\Gamma_\text{DD})+\epsilon$, where $\epsilon>0$ arbitrary small and the constraint over $P_{fa}^n$ is satisfied for $n$ large enough. See \cite[Prop. 2]{PoorFreqDet} for a detailed proof. This case allows to find the optimality of the PCS-MAC and PFS-MAC strategies among all orthogonal strategies for a fixed energy profile (see \cite[Th. 2]{Bianchi_2011}) and then, to obtain a closed form solution for the optimal energy profile $\xi(\nu)$. In this respect, 
we use variational calculus to solve the problem: 
\begin{align} %
\sup \kappa_m(\xi) \
\text{s.t. } &\xi(\nu)\geq 0\; \forall \nu \in\Theta_\beta,\label{eq:const1}\\
&\int_{\Theta_\beta}\xi(\nu)d\nu\leq c \label{eq:const2}
\end{align}
where $c=E_t/(\sigma^2_s+\sigma^2_v)$. 
Because the error exponent $\kappa_m(\xi)$ is an increasing function of $\xi$ (easily checked by computing the  fist derivative of its integrand), it is both a quasiconvex and quasiconcave (and thus quasilinear) function. The domain of the functional $\kappa_m(\xi)$ is all nonnegative functions, a convex set. The constraints of the problem are affine functions of $\xi$. Therefore, we have a quasiconvex optimization problem where the solution is not unique.
The Lagrangian is
\begin{align*}
L(\xi,\lambda,\nu)\!=\!  \int_{\Theta_\beta} \!\!\!\! \left\{-I(\xi(\nu)) - m(\nu)\xi(\nu)+ \lambda(\xi(\nu) -{c}/{\beta}) \right\}d\nu
\end{align*}
where $I(\xi(\nu))\!=\!\frac{\xi(\nu)\sigma^2_v + \sigma^2_w}{\xi(\nu)(\phi(\nu) + \sigma^2_v) + \sigma^2_w} +\log\frac{\xi(\nu)(\phi(\nu) + \sigma^2_v) + \sigma^2_w} {\xi(\nu)\sigma^2_v + \sigma^2_w} -1$ is the integrand of $\kappa_m(\xi)$. The scalar function $m(\nu)$ and the scalar constant $\lambda$ are the multipliers of Lagrange. 

The constraints (\ref{eq:const1}) and (\ref{eq:const2}) together with the following constraints are the Karush-Kuhn-Tucker necessary conditions for local extremes:  
\begin{itemize}
\item[$C_3)$] $\lambda^*-I'(\xi^*(\nu))\geq 0,\,\forall\nu\in\Theta_\beta$,
\item[$C_4)$] $[\lambda^*-I'(\xi^*(\nu))]\xi^*(\nu)=0,\,\forall\nu\in\Theta_\beta$,
\item[$C_5)$] $\lambda^*\geq 0$, 
\item[$C_6)$] $\lambda^* \int_{\Theta_\beta} (\xi^*(\nu)-c/\beta)d\nu =0$.
\end{itemize}
The Euler-Lagrange equation together with the complementarity condition of $m(\nu)$ and the solution to the problem $\xi^*(\nu)\defeq \xi_\text{OEP}(\nu)$ give the constraint $(C_4)$. The non-negativeness of $m(\nu)$ and $\lambda^*$ produce $(C_3)$ and $(C_5)$, respectively. Finally, $(C_6)$ is due to the complementarity of $\lambda^*$ and (\ref{eq:const2}).

$\xi(\nu)=0, \forall \nu$  and $\lambda=0$ satisfies all constraints but it is not the desired solution because the error exponent is 0. Then, for a non-trivial solution $\lambda^*>0$. It can be shown that the solutions to $I'(\hat{\xi})=\lambda^*$ correspond to the cubic equation (\ref{cubic}) with coefficients (\ref{eq:cubicCoeff}). Descartes' rule of signs of a polynomial establishes that then number of positive roots of a polynomial is related with the number of sign changes of the nonzero coefficients of consecutive powers. In the case of (\ref{eq:cubicCoeff}) and considering $(C_5)$, Descartes' rule determines that there could be at most 2 or 0 positive roots. Then, there exists a root, with a non-positive value, that does not satisfy (\ref{eq:const1}) and it is discarded.   

Define $\mathcal{I}_i=\{\nu\in[0,1]: \hat{\xi}_i(\nu)> 0\}$, $i=1,2$ as the sets of frequencies corresponding to the positive roots. Because of Descartes' rule, $\mathcal{I}_1\cap\mathcal{I}_2=\mathcal{I}_1=\mathcal{I}_2=\Theta_{\beta*} \subseteq \Theta_{\beta}$. We have two cases. If $\nu\notin\Theta_{\beta*}$, both $\hat{\xi}_1(\nu)$ and $\hat{\xi}_2(\nu)$ are non-positive and (\ref{eq:const1}) together with $(C_4)$ imply that $\xi^*(\nu)=0$. If $\nu\in\Theta_{\beta*}$, both $\hat{\xi}_1(\nu)$ and $\hat{\xi}_2(\nu)$ are positive and they are possible solutions. In fact, since $I(\xi(\nu))$ is an increasing function, $\xi^*(\nu)= \max\{\hat{\xi}_1(\nu), \hat{\xi}_2(\nu)\}$. In summary, $\xi^*(\nu)= \max\{\hat{\xi}_1(\nu), \hat{\xi}_2(\nu),\hat{\xi}_3(\nu),0\}$. A closed form for $\hat{\xi}^*(\nu)$ is shown in (\ref{xiSol}). $\lambda^*>0$ in $(C_6)$ implies that the energy constraint saturates, i.e., $\int_{\Theta_\beta}\xi^*(\nu)d\nu=c$ and this determines the value of $\lambda^*$. Finally, if $\nu\notin \Theta_{\beta*}$, $\xi(\nu)=0$ and $I(\xi(\nu))=0$. Then, we only need to integrate  $I(\xi(\nu))$ in $\Theta_{\beta*}$ to obtain the miss error exponent (\ref{eq:DDMiss}).

\bibliographystyle{IEEEtran}
\bibliography{IEEEabrv,../Refs/refs}

\end{document}